\documentclass[final]{l4dc2026}

\usepackage{amsmath,amssymb, amsfonts, mathrsfs}
\usepackage{graphicx}
\usepackage{xcolor}
\usepackage{comment}
\usepackage{dsfont}
\usepackage{csquotes}
\usepackage{caption, subcaption}

\usepackage{enumitem}
\usepackage{siunitx}
\usepackage{dsfont}
\usepackage{float}
\usepackage{dsfont}
\usepackage{wrapfig}
\usepackage{longtable}

\usepackage{algorithm} 

\newtheorem{assumption}[theorem]{Assumption}

\newcommand{\ie}{\textit{i.e., }}
\newcommand{\eg}{\textit{e.g., }}
    
\newcommand{\R}{\mathbb R}    
   
\newcommand{\gE}{\mathcal{E}}   
\newcommand{\gX}{\mathcal{X}}   
\newcommand{\gY}{\mathcal{Y}}

\DeclareMathOperator*{\argmin}{arg\,min}

\newcommand{\state}{X}

\newcommand{\sprob}{\Psi}

\newcommand{\mP}{\mathbb{P}}
\newcommand{\mE}{\mathbb{E}}

\usepackage[normalem]{ulem}
\usepackage{booktabs}
\usepackage{tikz}
\usepackage{pgfplots}
\usetikzlibrary{pgfplots.statistics}
\usepackage{tikzviolinplots}
\pgfplotsset{compat=1.18}
\usepackage{xcolor,soul}
\usepackage{verbatim}
\usepackage{multirow}
\usepackage{colortbl}
\usepgfplotslibrary{groupplots,dateplot}
\usepgfplotslibrary{fillbetween}
\usetikzlibrary{patterns,shapes.arrows}
\pgfplotsset{compat=newest}
\usetikzlibrary{shapes.arrows}
\usetikzlibrary{arrows.meta} 
\usetikzlibrary{positioning, fit, shapes.geometric}
\usetikzlibrary{backgrounds}
\usepackage{mathtools}
\usepackage{float} 
\definecolor{newblue}{RGB}{189,215,239} 
\definecolor{newgray}{RGB}{214,214,214} 
\definecolor{neworange}{RGB}{241,205,177} 
\definecolor{newyellow}{RGB}{255,193,19} 
\definecolor{newgreen}{RGB}{202,223,184} 

\usepackage{listings}
\usepackage{pdflscape}

\title[Adaptive PSC with Language Guidance]{Online Adaptive Probabilistic Safety Certificate with Language Guidance}
\usepackage{times}

\author{%
 \Name{Zhuoyuan Wang}\thanks{Equal contribution.} \Email{zhuoyuaw@andrew.cmu.edu}\\
 \addr Carnegie Mellon University
 \AND
 \Name{Xiyu Deng}\footnotemark[1] \Email{xiyud@andrew.cmu.edu}\\
 \addr Carnegie Mellon University
 \AND
 \Name{Hikaru Hoshino}\footnotemark[1] \Email{hoshino@eng.u-hyogo.ac.jp}\\
 \addr University of Hyogo
 \AND
 \Name{Yorie Nakahira} \Email{yorie@cmu.edu}\\
 \addr Carnegie Mellon University%
 \vspace{-1em}
}

\begin{document}

\maketitle

\vspace{-1em}

\begin{abstract}%

Achieving long-term safety in uncertain/extreme environments while accounting for human preferences remains a fundamental challenge for autonomous systems. Existing methods often trade off long-term guarantees for fast real-time control and cannot adapt to variability in human preferences or risk tolerance. To address these limitations, we propose a language-guided adaptive probabilistic safety certificate (PSC) framework that guarantees long-term safety for stochastic systems under environmental uncertainty while accommodating diverse human preferences. The proposed framework integrates natural-language inputs from users and Bayesian estimators of the environment into adaptive safety certificates that explicitly account for user preferences, system dynamics, and quantified uncertainties. Our key technical innovation leverages probabilistic invariance---a generalization of forward invariance to a probability space---to obtain myopic safety conditions with long-term safety guarantees. 
We validate the framework through numerical simulations of autonomous lane-keeping with human-in-the-loop guidance under uncertain and extreme road conditions, demonstrating enhanced safety–performance trade-offs, adaptability to changing environments, and personalization to different user preferences.
Code is available at \url{https://github.com/hoshino06/adaptive_lane_keeping}.


\end{abstract}


\section{Introduction}

Ensuring long-term safety in autonomous systems operating under uncertain and evolving environments is a fundamental challenge. In many real-world scenarios such as autonomous driving, both the environment and the human user’s preferences can vary over time. For instance, a driver may prefer specific driving styles and comfort levels while the road friction and traffic conditions remain unknown. Despite such variability, safety must always remain the top priority. This motivates the need for adaptive safe control frameworks that can efficiently adjust to diverse human preferences and uncertain dynamics, while maintaining provable long-term safety guarantees.
However, existing safe control methods face several limitations. First, most existing approaches lack systematic mechanisms to account for language guidance and users' preferences and risk tolerance, leading to suboptimal performance and alignment. Second, approaches based on adaptive control barrier functions often trade off long-term safety guarantees for real-time efficiency~\citep{taylor2020adaptive,xiao2021adaptive}, while probabilistic reachability-based techniques typically require expensive online computation to remain valid under environment uncertainty~\citep{abate2008probabilistic,chapman2019risk,kariotoglou2013approximate}. 

To address these challenges, we propose a novel language-guided adaptive probabilistic safety certificate (PSC) framework, which ensures long-term safety of stochastic systems in uncertain and adaptive settings while accounting for evolving human specifications and preferences. The overall framework is illustrated in Fig.~\ref{fig:diagram}. Our method enforces forward invariance in the probability space to guarantee long-term safety, achieved by integrating a system parameter estimator that updates the uncertainty representation with a safety certificate applying invariance conditions based on the current estimate. We derive a sufficient safety condition that guarantees long-term safety remains above a user-specified threshold, and show that this condition yields one-step control constraints that can be seamlessly incorporated into optimization-based controllers including model predictive control (MPC) with arbitrary planning horizons.
Unlike conventional control barrier function (CBF) approaches that operate in the state space and rely on hand-crafted barrier functions, the proposed adaptive PSC requires only a characterization of the desired safe set, enabling multi-turn large language models (LLMs)~\citep{yi2024survey,zhang2025survey} to effectively translate natural language instructions into formal safety specifications without compromising control performance. This design facilitates efficient online adaptation to evolving user guidance while ensuring long-term probabilistic safety under model uncertainty.
The main contributions of the paper are as follows. 

\begin{figure}[t!]
    \centering
    \includegraphics[width=\columnwidth]{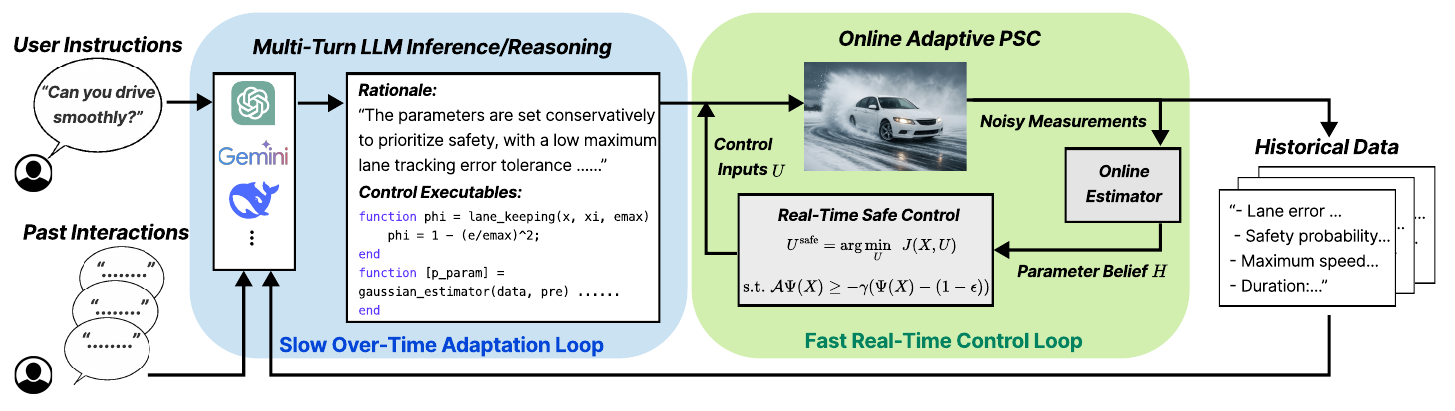}
    \vspace{-1.2em}
    \caption{The proposed language-guided adaptive probabilistic safety certificate (PSC) framework.
    }
    \vspace{-1em}
    \label{fig:diagram}
\end{figure}

\begin{itemize}
[leftmargin=12pt, itemsep=1pt, topsep=1pt, parsep=1pt, partopsep=1.5pt]
    \item We propose a novel adaptive probabilistic safety certificate that enforces forward invariance in the probability space, to systematically ensure long-term safety of uncertain systems with user-desired risk levels (Theorem~\ref{thm:main_theorem_aug}).

    \item We develop a holistic multi-turn LLM-based adaptive safe control framework that can translate natural language guidance into formal safety specifications, and adapt to changing environments and user preferences while maintaining long-term safety (Fig.~\ref{fig:diagram},~\ref{fig:llm_qualitative} and Table~\ref{tab:llm_case_1},~\ref{tab:llm_case_2}).

    \item We demonstrate the effectiveness of the proposed method in autonomous lane-keeping scenarios under uncertain and extreme road conditions. The results show significantly improved safety vs. efficiency trade-offs, effective adaptation to evolving human preferences, and sustained long-term safety with reduced online computational cost (Fig.~\ref{fig:horizon_vs_performance},~\ref{fig:tradeoff} and Table~\ref{tab:llm_case_1},~\ref{tab:llm_case_2}). 
\end{itemize}

\vspace{-1em}

\section{Related Work}

\textbf{Adaptive Safe Control.}
A wide range of safe control techniques has been developed for systems subject to noise and uncertainties when the safety specifications are given and fixed.   
For example, robust control barrier functions (CBFs) have been proposed to handle worst-case bounded disturbances in deterministic systems~\citep{cosner2021measurement,jankovic2018robust,Prajna2007}. To further address model uncertainties, adaptive CBFs~\citep{taylor2020adaptive,xiao2021adaptive} and their robust~\citep{lopez2020robust, castaneda2021pointwise, long2022safe} and high-order~\citep{xiao2021high,cohen2022high,nguyen2016exponential} variants have been introduced. Many of these methods consider bounded disturbances and worst-case framework, and may not be immediately applicable for stochastic systems. 

\noindent For stochastic systems, methods such as probabilistic CBFs~\citep{luo2020multi,Lyu2021,fan2020bayesian} and conditional value-at-risk (CVaR)-based approaches~\citep{ahmadi2021risk} have been developed. 
On the other hand, stochastic reachability or probabilistic reachability techniques are developed to estimate the likelihood that the system trajectory remains within safe regions or reaches desired targets in the presence of unbounded noise and uncertainties~\citep{fisac2018general, abate2008probabilistic,chapman2019risk,kariotoglou2013approximate,abate2006probabilistic,vasileva2020probabilistic}.
These techniques are often subject to a stringent tradeoff between assuring long-term safety vs. reducing real-time computation and may not be easily adapted to changing environments and human preferences~\citep{wang2026myopically}. In this paper, we allow the safe sets to depend on language guidance from human users, and the safety certificates to adapt to changing environments.


\noindent Model predictive control (MPC) methods have also been widely studied for decision-making under uncertainty, including chance-constrained, stochastic, and probabilistic MPC formulations~\citep{farina2016stochastic,heirung2018stochastic,mesbah2016stochastic}. The proposed safety certificate can be seamlessly incorporated into MPC as additional constraints. Unlike conventional MPC approaches, where the planning horizon must exceed the safety constraint horizon to maintain long-term safety, the proposed safety certificate can enable longer-term safety with shorter MPC horizons. With augmented risk information, the proposed method can decouple the computational complexity of MPC from the long-term safety horizon (see Fig.~\ref{fig:horizon_vs_performance}-\ref{fig:trajectory_H20}).


\noindent \textbf{Language-Guided Control.} 
Large language models (LLMs) have demonstrated the effectiveness of multi-turn prompting for generating contextually consistent and coherent responses through iterative refinement~\citep{yi2024survey, zhang2025survey}. To incorporate human intent and preferences into control, LLMs have been combined with model predictive control (MPC) frameworks to interpret and integrate natural-language guidance into control objectives~\citep{maher2025llmpc,miyaoka2024chatmpc,sanyal2025asma}.  
LLMs have also been used to generate context-aware behaviors~\citep{deng2024context,khan2025safety,wu2025selp} or to couple language reasoning with reachability analysis for safe motion planning~\citep{bayat2025llm,hafez2025safe}. These approaches highlight the potential of LLMs as adaptive reasoning modules that translate high-level language inputs into control-relevant information.

\noindent Recent studies have explored LLM-guided safe control, where language models assist in constructing control barrier functions (CBFs) and defining corresponding safety conditions for deterministic systems~\citep{khan2025safety,sanyal2025asma,brunke2025semantically}. However, since multiple CBF formulations can satisfy the same safety requirement yet induce markedly different control behaviors, these approaches often rely on manual design of CBFs by human engineers or unreliable automatic generation of barrier functions by LLMs. 
In contrast, the proposed adaptive probabilistic safety certificate (PSC) framework leverages LLMs solely to generate high-level safe set specifications, which are subsequently translated into formal safety conditions using stochastic system models and Bayesian estimators. This yields a principled and reliable process for constructing safety certificates that ensure long-term safety and alignment in stochastic, non-stationary environments.

\section{Problem Formulation}
\label{sec:problem_formulation}

\textbf{Preliminary.} 
Let $f:\gX \rightarrow \gY$ represent that $f$ is a mapping from space $\gX$ to space $\gY$. Let \(\mathds{1}\{\gE\}\) be an indicator function, which takes \(1\) when condition \(\gE\) holds and \(0\) otherwise. Given events $\gE$ and $\gE_c$, let $\mP(\gE)$ be the probability of $\gE$ and $\mP(\gE | \gE_c)$ be the conditional probability of $\gE$ given the occurrence of $\gE_c$. Given random variables $X$ and $Y$, let $\mE[X]$ be the expectation of $X$ and $\mE[X | Y = y]$ be the conditional expectation of $X$ given $Y=y$. We use upper-case letters (\eg $Y$) to denote random variables and lower-case letters (\eg $y$) to denote their specific realizations.

We consider adaptive safe control problems with language guidance, where human user provides task-related commands and preferences, and an adaptive controller is designed to follow human instructions, while maintaining \textit{long-term} safety of the system despite system uncertainties. 
Specifically, we consider a discrete-time stochastic control system with unknown parameters and human-specified safety requirements. The system evolves according to  
\vspace{-2pt}
\begin{align}
    X_{k+1} \sim F(dX_{k+1} \mid X_k, U_k; \xi),  
    \label{eq:discrete_dynamics}
\end{align}
where $X_k \in \mathcal{X} \subset \mathbb{R}^n$ denotes the system state, $U_k \in \mathcal{U} \subset \mathbb{R}^m$ is the control input, and $\xi \in \mathbb{R}^l$ is an \emph{unknown system parameter}. The transition kernel 
$F: \mathcal{B}(\mathcal{X}) \times \mathcal{X} \times \mathcal{U} \times \mathbb{R}^l \to [0,1]$ 
characterizes the stochastic dynamics, with $\mathcal{B}(\mathcal{X})$ denoting the Borel $\sigma$-algebra of $\mathcal{X}$~\citep{hernandez2012discrete}.
A human-specified safety requirement is described by a function $\phi: \mathbb{R}^n \times \mathbb{R}^l \to \mathbb{R}$, and the corresponding safe set is defined as
\begin{align}
   \mathcal{C}(\xi) := \{ x \in \mathbb{R}^n : \phi(x, \xi) \ge 0 \},
   \label{eq:safe_set}
\end{align}
where $\phi$ captures safety constraints. 
While prior works manually design $\phi$ to improve safe control performance~\citep{ames2019control}, our approach is agnostic to different yet equivalent representations of $\phi$, enabling seamless integration with large language models (LLMs) for automating its construction from human specifications of safety requirements.\footnote{See Appendix~\ref{sec:llm_generated_cbf} for ablation studies demonstrating consistent performance of the proposed method across diverse forms of (LLM-generated) barrier functions.}
We then define an indicator function for long-term safety over a finite horizon $\tau_k = \{k, k+1, \dots, k+T\}$ as 
\vspace{-5pt}
\begin{align}
 \mathds{1}_\mathrm{LS}( X_{k:k+T}, \xi ) :=  
 \begin{cases}
     1, & \text{if}~ X_{k'} \in \mathcal{C}(\xi),~ \forall k' \in \tau_k, \\
     0, & \text{otherwise},
     \vspace{-5pt}
 \end{cases}
\end{align}
where $X_{k:k+T} := \{ X_k, X_{k+1}, \dots, X_{k+T} \}$ denotes the trajectory segment.

To capture epistemic uncertainty in the unknown parameter~$\xi$, we define a sequence of random variables $\{\Xi_i\}_{i \in \mathbb{Z}_+}$ with corresponding probability density functions $\{H_i\}_{i \in \mathbb{Z}_+}$, each reflecting the current estimate of the uncertain parameter at time step~$i$. These distributions characterize uncertainty arising from limited or noisy measurements and can be estimated using many estimators, including Bayesian neural networks~\citep{hafner2018reliable, gal2016dropout} and Gaussian processes~\citep{shahriari2015taking, pan2017prediction}. The estimator and its configuration can be determined based on human specifications, encoding user preferences and desired adaptation behavior.
Then, for a given nominal controller~$\pi: \mathcal{X} \times \mathbb{R}^l \rightarrow \mathcal{U}$ and parameter belief~$H_i$, the long-term safety probability from state~$x$ is defined as
\begin{align}
   \Psi^{\pi}_{H_i}(x) 
   &= \mathbb{E}_{\Xi_i \sim H_i}\big[ \mathds{1}_\mathrm{LS}(X_{k:k+T}, \Xi_i) \mid X_k = x \big] \notag\\
   &= \int_{\mathcal{B}(\mathcal{X})^T \times \mathbb{R}^l}
      \mathds{1}_\mathrm{LS}(X_{k:k+T}, \Xi_i)
      \mathbb{P}(dX_{k+1:k+T} \mid X_k = x, \pi, \Xi_i)
      H_i(\Xi_i) \, d\Xi_i,
   \label{eq:safety_prob}
\end{align}
where the control inputs $U_k$ is computed based on the policy $\pi$ and the trajectory distribution is given by
\vspace{-10pt}
\begin{align}
   \mathbb{P}(dX_{k+1:k+T}\mid X_k=x, \Xi_i)
   = \prod_{k'=k}^{k+T-1}
     F(dX_{k'+1} \mid X_{k'}, \pi(X_{k'}, \hat{\xi}_i); \Xi_i).
\end{align}
Here, $\hat{\xi}_i = \int_{\mathbb{R}^l} \Xi_i H_i(\Xi_i) \, d\Xi_i$ represents the current estimate of the unknown parameter, which is used in the nominal controller $\pi(X_k, \hat{\xi}_i)$.  
Since the dynamics~\eqref{eq:discrete_dynamics} is time-invariant, $\Psi^{\pi}_{H_i}(x)$ does not depend on $k$.

The goal is to design an adaptive safe control policy $\pi_{\text{safe}}$ that updates both the control and estimation strategy to ensure that, at every time step, the expected long-term safety probability maintains higher than a risk tolerance threshold $1-\epsilon$, \ie 
\begin{align}
    \mathbb{E}_{\pi_{\text{safe}}}\big[ \Psi^{\pi}_{H_k}(X_k) \big] \ge 1 - \epsilon,
    \quad \forall k \in \mathbb{Z}_+.
    \label{eq:safe_control_objective}
\end{align}
The expectation is taken over the distribution of $X_k$ conditioned on the initial state $X_0$ with the closed-loop policy $U_t = \pi_{\text{safe}}(X_t, \hat{\xi}_t)$ for $t \in \{k, \dots, k+T\}$.  
In addition to maintaining long-term safety~\eqref{eq:safe_control_objective}, the controller minimizes a performance objective $J(X_k, U_k)$, which may vary based on user preferences.
For instance, one may consider $J(X_k, U_k) = \mathbb{E} [\sum_{t=k}^{k+T} c(X_t, U_t)]$, where $c(\cdot, \cdot)$ is a cost function on states and controls. 
This formulation captures the interaction between human-specified safety requirements, estimator preferences, and performance objectives under a unified adaptive control framework that guarantees long-term safety in stochastic systems with unknown parameters. See~\cite{wang2026myopically} for a detailed discussion of the advantages of the long-term safety formulation compared to existing stochastic safe control methods, such as~\cite{luo2020multi}.


\section{Proposed Method}
\label{sec:proposed_method}

In this section, we present the proposed safe control framework with language guidance. The overall architecture is illustrated in Fig.~\ref{fig:diagram}. The framework consists of two interconnected loops: an inner safe control loop, which operates in \textit{real time} based on the adaptive probabilistic safety certificate (PSC), and an outer multi-turn LLM loop, which evolves at a \textit{slower timescale} to incorporate human instructions and historical control data. The inner loop ensures long-term safety under uncertainty and diverse human-specified requirements, while the outer loop enables the system to adapt to human preferences and correct potentially erroneous information over time. In the following, we present the proposed adaptive PSC in Section~\ref{sec:adaptive_PSC_method} and its integration with the LLM in Section~\ref{sec:llm_integration}.

\subsection{Adaptive Probabilistic Safety Certificate}
\label{sec:adaptive_PSC_method}

In this section, we introduce the adaptive probabilistic safety certificate that guarantees long-term safety despite uncertainties in the dynamics.
We start by defining a discrete-time generator $\mathcal{A}^{U_k}_{H_{k+1}}$ for the safety probability, which depends on the control input $U_k$ and the density function $H_{k+1}$ for the next time step:
\vspace{-2pt}
\begin{align}
\label{eq:generator_A}
    \mathcal{A}^{U_k}_{H_{k+1}} \Psi^{\pi}_{H_k} (X_k) & := \dfrac{ \mathbb{E}[\Psi^\pi_{H_{k+1}}(X_{k+1}) | X_k, U_k] - \Psi^\pi_{H_k}(X_k) }{\Delta t},
\end{align}
where $\Delta t$ is the sampling time interval. 
Then, to ensure the long-term safety of the system, we propose to constrain the control action $U_k$ to satisfy the following conditions at all time $k\in\mathbb{Z}_+$: 
\begin{align}
\label{eq:safety_condition_each_Zt_aug}
    \mathcal{A}^{U_k}_{H_{k+1}} \Psi^\pi_{H_k}(X_k) \geq -\gamma(\Psi^\pi_{H_{k}}(X_k)-(1-\epsilon)),
\end{align}
where $\gamma: \R \rightarrow \R$ is a function that satisfies the following requirements:
\begin{itemize}
    \item[] Requirement 1: $\gamma(q)$ is strictly concave or linear and strictly increasing in $q$.
    \item[] Requirement 2: $\gamma(q) \leq q$, ~$\forall q \in \mathbb{R}$.
\end{itemize}

\begin{remark}
    In practice, choose $\gamma(q) = aq$ where $a \in [0, 1]$ will satisfy the two requirements. The proposed safety condition~\eqref{eq:safety_condition_each_Zt_aug} is a control constraint on $U_k$. This constraint is linear in $U_k$ if the system dynamics~\eqref{eq:discrete_dynamics} is control affine~\citep{wang2022myopically, wang2026myopically}. Numerical methods for evaluating~\eqref{eq:generator_A} are provided in Appendix~\ref{sec:appendix_psc_calculation}. 
\end{remark}

\begin{theorem}
\label{thm:main_theorem_aug}
Consider the stochastic dynamics given by \eqref{eq:discrete_dynamics} with a nominal controller $\pi$ and a sequence of probability density functions $H_0, H_1, \dots$ for estimating $\xi$. 
Suppose the state and parameter estimation originate from $X_0$ and $H_0$ satisfy the following
\begin{align}
  \Psi^\pi_{H_{0}}(X_0) > 1 - \epsilon,
\end{align}
and the control action $U_k$ given by the controller $\pi_\mathrm{safe}$ satisfies \eqref{eq:safety_condition_each_Zt_aug} at all time.
Then, the following condition holds:
\vspace{-8pt}
\begin{align}
  \label{eq:satisfy_control_policy_aug}
    \mathbb{E}_{\pi_\mathrm{safe}}[ \Psi^\pi_{H_k}(X_k) ] \geq 1-\epsilon, 
    \quad \forall k \in \mathbb{Z}^+. 
\end{align}
\end{theorem}
\vspace{-3pt}
\begin{proof}
    See Appendix~\ref{sec:appendix_proof}.
\end{proof}

Theorem~\ref{thm:main_theorem_aug} says that given the initial state to be sufficiently safe, with the constraint of the safety certificate~\eqref{eq:safety_condition_each_Zt_aug}, the resulting policy $\pi_\mathrm{safe}$ ensures long-term safety in the sense that the expected long-term probability $\Psi^\pi_{H_k}(X_k)$ is always above the desired tolerance.
With these theoretical guarantees, the safe controller $\pi_{\text{safe}}$ can be calculated at each time step using the following optimization.
\vspace{-2pt}
\begin{align}
\label{eq:conditioning}
    \begin{aligned}
        U_k^{\text{safe}}(X_k, H_k) = && \argmin_{U} & \ \ J(\pi(X_k, \hat{\xi_k}),U) \\
        && \text{s.t.} & \ \ \eqref{eq:safety_condition_each_Zt_aug},
    \end{aligned}
\end{align}
where $\pi$ is the reference controller\footnote{The reference controller $\pi$ is usually set to be the nominal controller $\pi$ to estimate safety probability, but the theory does not restrict the two to be the same.}, and $J$ is a generic cost function that characterizes human's preference on the control performance. (\eg quadratic cost on the distance between the safe control and a performance-oriented nominal control). Here $\hat{\xi_k} = \int_{\R^l} \Xi_k H_k(\Xi_k) d\Xi_k$ is the current estimate of the unknown parameter at step $k$. 
Similarly, the proposed adaptive safety certificate constraint~\eqref{eq:safety_condition_each_Zt_aug} can also be integrated to the model predictive control (MPC) framework as below.
\vspace{-3pt}
\begin{align}
\label{eq:conditioning_mpc}
    \begin{aligned}
        U_k^{\text{safe}}(X_k, H_k) = && \argmin_{U_k} & \ \ J(X_{k:k+T_{\mathrm{mpc}}}, U_{k:k+T_{\mathrm{mpc}}}) \\
        && \text{s.t.} & \ \ \eqref{eq:discrete_dynamics} , \quad \forall t \in [k, k+T_{\mathrm{mpc}}]\\
        &&&  \ \ \eqref{eq:safety_condition_each_Zt_aug}, \quad t = k
    \end{aligned}
\end{align}
where $T_{\mathrm{mpc}}$ is the MPC planning horizon. Note that the long-term safety constraint~\eqref{eq:safety_condition_each_Zt_aug} only needs to be enforced at the current time step $k$, and the long-term safety horizon $T$ can be different than the MPC horizon $T_{\mathrm{mpc}}$. This offers real-time computational advantages and strong safety guarantees with a short-horizon MPC, as demonstrated in the experiments in Section~\ref{sec:adaptive_mpc_exp}.
The overall adaptive safe control algorithm is shown in Appendix~\ref{sec:appendix_psc_alg}.

\subsection{Integration with Language Guidance}
\label{sec:llm_integration}


We integrate large language models (LLMs) with the proposed adaptive probabilistic safety certificate (PSC) to enable adaptation to human specifications and preferences while ensuring long-term safety under environmental uncertainties. The overall LLM integration pipeline is illustrated in Fig.~\ref{fig:diagram} (highlighted in blue).
The LLM adaptation loop operates on a slower timescale than the online adaptive safe control and therefore does not require real-time inference or reasoning. At initialization, the user provides natural language instructions describing the desired control behavior. Through a system-level prompt, an inference LLM interprets these instructions and generates rationales and control executables, such as safe set designs, estimator priors, and estimator update rules, which are passed to the adaptive PSC for control execution. Historical data on control performance, safety satisfaction, and estimator updates are collected and, in subsequent runs, provided together with the user’s latest instruction to a reasoning LLM that produces updated executables for control. By jointly leveraging human inputs and accumulated data, the multi-turn LLM can automatically correct erroneous human information and improve control performance, as demonstrated in Section~\ref{sec:exp_llm_mpc}. Implementation details of the LLM integration, including system-level prompts and algorithmic procedures, are provided in Appendix~\ref{sec:appendix_llm_implementation}.

\section{Case Study: Adaptive Lane Keeping}
\label{sec:experiment}

\subsection{Problem Setting}

We consider a challenging lane-keeping problem to illustrate the effectiveness of the proposed method.  
Whereas lane-keeping problems are usually addressed under the assumptions of a constant longitudinal velocity and a linear tire model~\citep{rossetter2006,ames2017control}, we consider a lane-keeping problem under extreme road conditions where the friction coefficient can be low and uncertain. 
In wet or icy road conditions, tire forces may saturate, and coupled longitudinal and lateral vehicle dynamics affect tire stability~\citep{YW2022}. 

To simulate these complex dynamics, we use a three-degree of freedom (3-DOF) vehicle model with the LuGre combined-slip tire model to capture nonlinearities under extreme conditions~\citep{Huang2017,li2023}. 
The system state consists of three components: the vehicle motion $x_\mathrm{vehicle}$, the wheel dynamics $x_\mathrm{wheel}$, and the vehicle position $x_\mathrm{road}$ relative to the road coordinate. Specifically, the system state is
\begin{align}
\label{eq:vehicle_state}
 x = 
 [ \,\underbrace{v_x, \, v_y, \, r, \delta}_{x_\mathrm{vehicle}}, \,\underbrace{\omega_\mathrm{fl}, \, \omega_\mathrm{fr}, \, \omega_\mathrm{rl},  \omega_\mathrm{rr}, \, \tau_\mathrm{e}}_{x_\mathrm{wheel}}, \,\underbrace{ s, \, e, \,\psi }_{x_\mathrm{road}} \, ]^\top.
\end{align}
Here, the state $x_\mathrm{vehicle}$ consists of the vehicle's longitudinal and lateral velocities $v_x$ and $v_y$, the yaw rate $r$, and the steering angle $\delta$. 
The state $x_\mathrm{wheel}$ consists of the wheel's angular velocities $\omega_\mathrm{fl}$, $\omega_\mathrm{fr}$, $\omega_\mathrm{rl}$ and $\omega_\mathrm{rr}$, and the torque $\tau_\mathrm{e}$, which is greater than or equal to 0 when accelerating and smaller than 0 when braking. 
The state $x_\mathrm{road}$ consists of the distance $s$ traveled from a reference point (starting point), the lateral error $e$ from the center line of the road, and the relative heading error $\psi$. 
The control action is $u = [\Delta \delta, \, \Delta \tau_\mathrm{e}  ]^\top$, where $\Delta \delta$ and $\Delta \tau_\mathrm{e}$ stand for the rates of change of $\delta$ and $\tau_\mathrm{e}$, respectively. The unknown parameter $\xi$ in the dynamics is the friction coefficient $\mu$.
The detailed vehicle dynamics model and parameters used can be found in Appendix~\ref{sec:appendix_vehicle_model}.

We formulate the control objective by using the safe set $\mathcal{C}$ defined by \eqref{eq:safe_set}, with 
\vspace{-3pt}
\begin{align}
\label{eq:lane_error_safe_set}
    \phi(x, \xi) =  1 - \left( \dfrac{e}{e_\mathrm{max}} \right)^2,
\end{align}
where $e_\mathrm{max}$ is the allowable limit of the absolute value of the lateral error $e$ and can have different values given different levels of conservativeness. We use a Bayesian estimator for unknown parameter update of $\xi$, and the mean $\mu_k$ and the variance $\sigma_k^2$ of the Gaussian posterior $P_k(\mu | H_k)$ are sequentially updated by the measurement $M_k$ of the friction coefficient as follows~\citep{murphy2007}:
\begin{align}
\label{eq:estimator_update}
  \mu_k &= \dfrac{ \bar{\sigma}^2\mu_{k-1}+\sigma_{k-1}^2 M_k }{ \bar{\sigma}^2+\sigma_{k-1}^2 }, \quad
  \sigma_k^2 = \dfrac{ \bar{\sigma}^2 \sigma_{k-1}^2}{ \bar{\sigma}^2 +\sigma_{k-1}^2 },
\end{align}
where $\bar{\sigma}^2$ is the variance of $M_k$, reflecting the measurement noise level.

\subsection{Long-Term Safety and Efficiency under Uncertainty}
\label{sec:adaptive_mpc_exp}

We first demonstrate that the proposed adaptive probabilistic safety certificate ensures long-term safety under diverse uncertain environments, and achieves better safety vs. efficiency trade-offs. Specifically, we investigate the integration of the adaptive PSC with model predictive control (MPC). Since standard control barrier function (CBF) constraints with~\eqref{eq:lane_error_safe_set} will result in infeasibility, we examine the following two baseline controllers: (i) Adaptive MPC (\textbf{AMPC}), a nonlinear MPC scheme employing the LuGre tire model to enhance vehicle stability under extreme road conditions~\citep{li2023}; and (ii) Control-Dependent Barrier Functions (\textbf{CDBF})~\citep{Huang2021,YW2022}, a variant of CBF that combines MPC with safety constraints to handle unstable tire forces and coupled longitudinal--lateral dynamics. 
MPC formulations and implementation details are provided in Appendix~\ref{sec:sub_appendix_general_setup} and online computation times are reported in Appendix~\ref{sec:sub_appendix_adaptive_psc}.

\begin{figure}[t!]
    \centering

    \begin{minipage}[b]{0.40\linewidth}
        \centering
        \includegraphics[width=\linewidth]{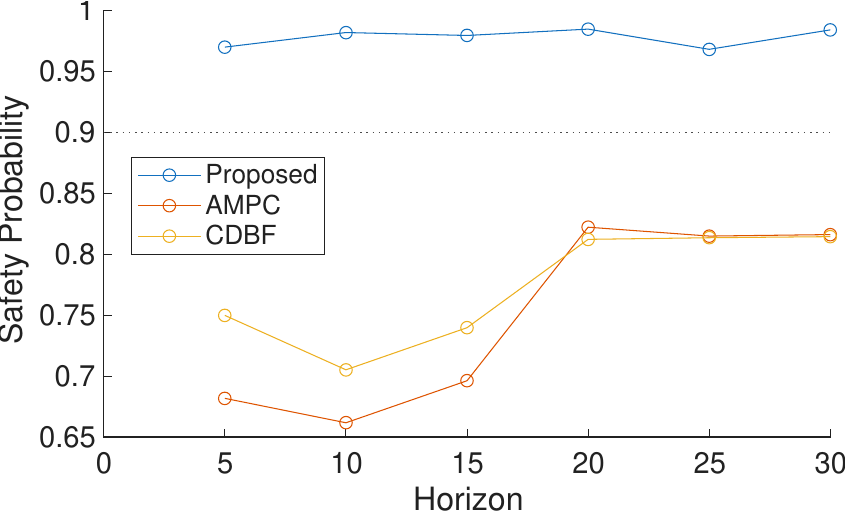}
        \caption{Safety probability vs. MPC horizon $T_{\mathrm{mpc}}$.}
        \label{fig:horizon_vs_performance}
    \end{minipage}
    \hfill
    \begin{minipage}[b]{0.25\linewidth}
        \centering
        \includegraphics[width=\linewidth]{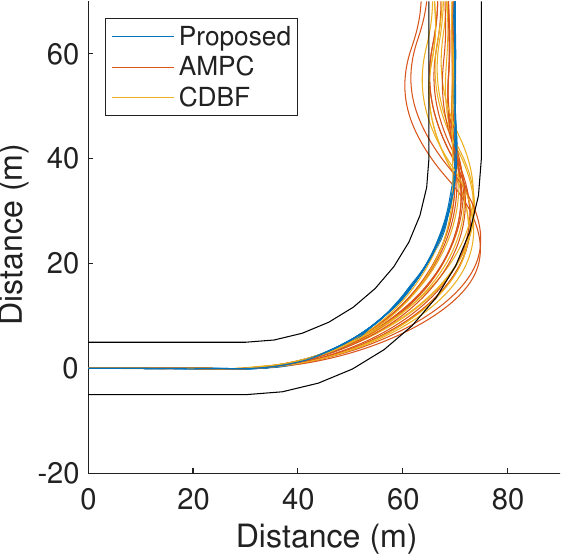}
        \caption{Vehicle trajectories ($T_{\mathrm{mpc}}=10$).}
        \label{fig:trajectory_H10}
    \end{minipage}
    \hfill
    \begin{minipage}[b]{0.25\linewidth}
        \centering
        \includegraphics[width=\linewidth]{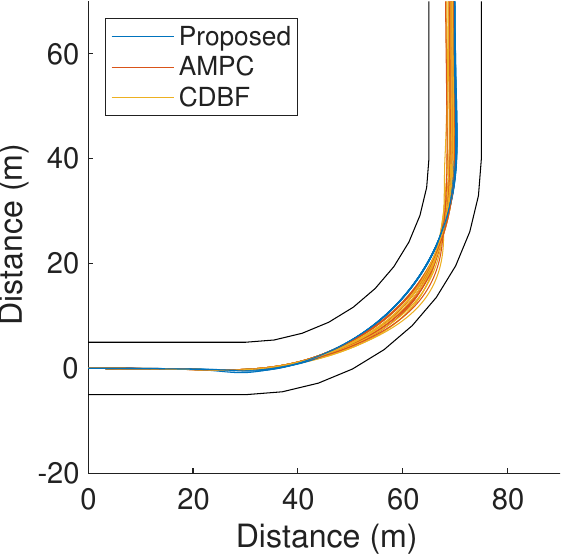}
        \caption{Vehicle trajectories ($T_{\mathrm{mpc}}=20$).}
        \label{fig:trajectory_H20}
    \end{minipage}
\end{figure}

We tested closed-loop performance through 30 simulations with icy road conditions (unknown ground truth friction coefficient $\mu \in [0.2, 0.4]$) and desired safety probability $1-\epsilon = 0.9$. 
Fig.~\ref{fig:horizon_vs_performance} shows the effect of the prediction horizon $T_{\mathrm{mpc}}$ on performance for each controller. 
It can be seen that with AMPC and CDBF, the safety probability is compromised when $T_{\mathrm{mpc}}$ is reduced, while with the proposed adaptive PSC, the safety probability remains above the desired threshold $1-\epsilon$ regardless of the MPC horizon. 
Fig.~\ref{fig:trajectory_H10} and~\ref{fig:trajectory_H20} show the vehicle trajectories with different MPC horizons $T_{\mathrm{mpc}}$. 
For both AMPC and CDBF, only state constraints are imposed inside the MPC optimization, thus lane keeping behavior will fail at $T_{\mathrm{mpc}}=10$ and can only be ensured with larger $T_{\mathrm{mpc}}=20$.
In contrast, the proposed controller achieves the desired safe behavior even with a shorter MPC horizon, as the probabilistic constraint ensures long-term safety with a myopic evaluation.

\begin{wrapfigure}{r}{0.42\linewidth}
    \vspace{-10pt}
    \centering
    \includegraphics[width=1.1\linewidth]{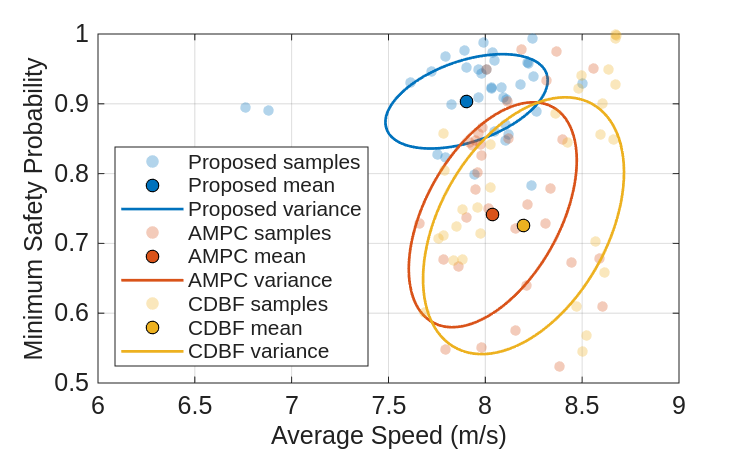}
    \caption{Safety vs. efficiency trade-offs with MPC-based adaptive controls. 
    }
    \label{fig:tradeoff}
    \vspace{-9pt}
\end{wrapfigure}

We evaluated the closed-loop performance of all three methods under varying road conditions with unknown ground truth friction coefficients: icy ($\mu \in [0.3, 0.4]$), wet ($\mu \in [0.5, 0.6]$), and dry ($\mu \in [0.8, 0.9]$), and we tested with different estimator priors, measurement noise levels, and maximum lane-error tolerances, to conduct a comprehensive quantitative analysis. Specifically, we consider MPC horizon $T_{\mathrm{mpc}}=10$, estimator prior means $\xi_0 = [0.3, 0.5, 0.9]$, prior standard deviations $\sigma_0 = [0.05, 0.3]$, measurement noise standard deviations $\bar{\sigma} = [0.05, 0.3]$, and maximum lane-error tolerances $e_\mathrm{max} = [3, 5, 10]$, yielding 108 distinct settings. For each feasible setting\footnote{See Appendix~\ref{sec:sub_appendix_adaptive_psc} for detailed definitions and results on both feasible and infeasible cases.} and each method, we run 20 simulations and visualize the average minimum safety probability and average vehicle velocity as trade-off plots in Fig.~\ref{fig:tradeoff}.
We can see that the proposed adaptive PSC method consistently achieves the desired safety probability, which is significantly higher than that of AMPC and CDBF, while maintaining comparable average vehicle velocities. These results show that the proposed method achieves better trade-offs between safety and efficiency, by ensuring long-term safety without compromising performance.



\subsection{Adaptation with Human Instructions}
\label{sec:exp_llm_mpc}

In this section, we demonstrate how the proposed framework integrates large language models (LLMs) with the adaptive probabilistic safety certificate (PSC) to enable safe and responsive adaptation in accordance with human preferences and specifications. 

We first present a qualitative analysis of the proposed framework in Fig.~\ref{fig:llm_qualitative}, with three rounds of human instructions and feedback under a driving scenario with icy road conditions. It can be seen that with the initial human specification and preference, the vehicle produces a smooth and slower ride. Then with the second round human feedback, the proposed model adapted to the human preferences and achieved more aggressive driving behavior while maintaining safety. 
In the final round, when the human feedback becomes irrelevant, the model reasons from historical data and past interactions to infer intent, producing smooth and efficient driving behavior that balances performance and safety. Note that in all three cases, safety is guaranteed despite evolving user preferences and uncertain environments, thanks to the proposed adaptive PSC.

We then present a quantitative analysis to show that the proposed framework can effectively adapt to human preferences and correct human errors based on historical data, while ensuring safety of the system under uncertainties. We consider two cases: (i) the human exhibits varying preferences over time, to which the framework adapts accordingly, (ii) the human provides erroneous judgments about road conditions, which the framework automatically corrects.
We collect diverse human instructions for different degrees of driving aggressiveness, safety preferences, and judgments about road conditions, with different levels of clarity. We then classify them into different categories and run batch experiments to examine intelligent adaptation and safety assurance. Table~\ref{tab:llm_case_1} shows the case where the human first requests an aggressive ride in Run~1 and then prefers a conservative ride in Run~2. We measure the center-lane lateral deviation (\textit{Lateral}), average velocity (\textit{Speed}), and empirical safety (\textit{Safety}) of the system. It can be seen that with the proposed framework, both the lane deviation and the average velocity are reduced according to human preference, while safety is maintained. Table~\ref{tab:llm_case_2} shows the case where the human incorrectly assumes a dry road when it is actually icy. Leveraging historical data including trajectory information and estimator \textit{Posterior}, the framework refines the \textit{Prior} estimate of the friction coefficient. Therefore, despite persistent erroneous human instruction, Run~2 achieves smoother trajectories compared to Run~1, while maintaining safety and efficiency.
We also evaluate robustness across different LLMs, including GPT-4o-mini, Gemini 2.5 Flash, and DeepSeek-Chat, demonstrating consistent performance. Full experiment results and ablations are provided in Appendix~\ref{sec:sub_appendix_llm_setup}.

\begin{figure}
    \centering
    \includegraphics[width=\columnwidth]{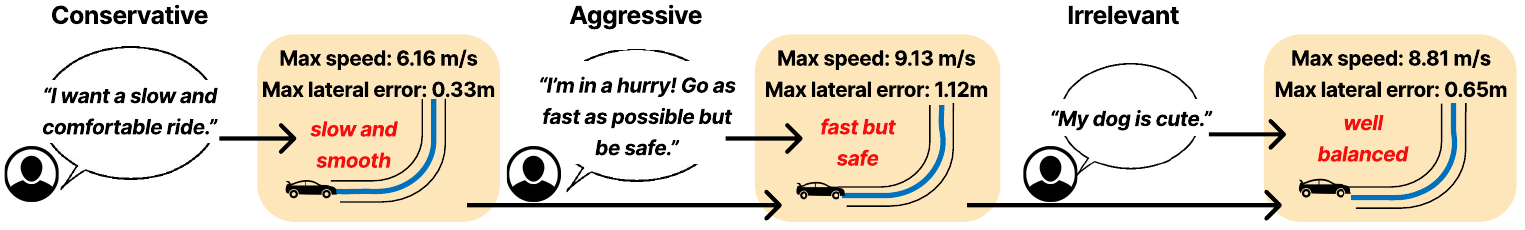}
    \vspace{-0.8em}
    \caption{Qualitative results across three rounds of human instructions and feedback.}
    \vspace{-0.3em}
    \label{fig:llm_qualitative}
\end{figure}

\begin{table}[t!]
\centering
\scriptsize
\begin{tabular}{lccc ccc cc}
\toprule
\textbf{Icy Road $\mu \in [0.3, 0.4]$} &
\multicolumn{3}{c}{Run~1: Aggressive Preference} &
\multicolumn{3}{c}{Run~2: Conservative Preference} &
\multicolumn{2}{c}{Differences} \\
\cmidrule(lr){1-1} \cmidrule(lr){2-4} \cmidrule(lr){5-7} \cmidrule(lr){8-9}
\textbf{LLM} &
\uline{\textit{Lateral}} & \uline{\textit{Speed}} & \uline{\textit{Safety}} &
\uline{\textit{Lateral}} & \uline{\textit{Speed}} & \uline{\textit{Safety}} &
\uline{\textit{Lateral}} & \uline{\textit{Speed}} \\
\midrule\addlinespace[2pt]\midrule
GPT-4o-mini & $0.78 \!\pm\! 0.65$ & $8.17 \!\pm\! 1.54$ & $99\%$
            & $0.55 \!\pm\! 0.38$ & $7.73 \!\pm\! 1.40$ & $100\%$
            & $-0.23$ & $-0.44$ \\
Gemini 2.5 Flash & $0.90 \!\pm\! 0.54$ & $8.43 \!\pm\! 1.01$ & $99\%$
                    & $0.39 \!\pm\! 0.36$ & $6.92 \!\pm\! 1.52$ & $96\%$
                    & $-0.51$ & $-1.51$ \\
DeepSeek-Chat & $0.87 \!\pm\! 0.55$ & $8.52 \!\pm\! 1.51$ & $100\%$
                 & $0.43 \!\pm\! 0.41$ & $7.17 \!\pm\! 1.64$ & $100\%$
                 & $-0.44$ & $-1.35$ \\
\bottomrule
\end{tabular}
\vspace{-5pt}
\caption{Quantitative results for case (i): preference adaptation based on human instructions.}
\label{tab:llm_case_1}
\end{table}


\begin{table}[t!]
\centering
\scriptsize
\begin{tabular}{lccc ccc cc}
\toprule
\textbf{Icy Road $\mu \in [0.3, 0.4]$} 
& \multicolumn{3}{c}{Run~1: Wrong Premise on Dry Road} 
& \multicolumn{3}{c}{Run~2: Wrong Premise on Dry Road} 
& \multicolumn{2}{c}{Differences} \\[2pt]
\cmidrule(lr){1-1} \cmidrule(lr){2-4} \cmidrule(lr){5-7} \cmidrule(lr){8-9}
\textbf{LLM} & \uline{\textit{Prior}} & \uline{\textit{Posterior}} & \uline{\textit{Safety}} 
& \uline{\textit{Prior}} & \uline{\textit{Posterior}} & \uline{\textit{Safety}} 
& \uline{\textit{Lateral}} & \uline{\textit{Speed}} \\ 
\midrule\addlinespace[2pt]\midrule
GPT-4o-mini & $0.82 \!\pm\! 0.16$ & $0.35 \!\pm\! 0.03$ & $94\%$
            & $ 0.58 \!\pm\! 0.21$ & $0.35 \!\pm\! 0.03$ & $100\%$
            & $-0.14$ & $-0.32$ \\ 
            
Gemini 2.5 Flash   & $0.90 \!\pm\! 0.04$ & $0.36 \!\pm\! 0.03$ & $100\%$
            & $0.56 \!\pm\! 0.26$ & $0.36 \!\pm\! 0.03$ & $100\%$
            & $-0.11$ & $-0.08$ \\ 
            
DeepSeek-Chat & $0.90 \!\pm\! 0.04$ & $0.35 \!\pm\! 0.03$ & $98\%$
            & $0.84 \!\pm\! 0.14$ & $0.36 \!\pm\! 0.03$ & $100\%$
            & $-0.42$ & $+0.04$ \\ 
\bottomrule
\end{tabular}
\vspace{-5pt}
\caption{Quantitative results for case (ii): correcting human error with historical data.}
\label{tab:llm_case_2}
\vspace{-1em}
\end{table}

\vspace{-0.6em}

\section{Conclusion}
\label{sec:conclusion}

\vspace{-0.1em}

In this paper, we proposed a language-guided adaptive probabilistic safety certificate (PSC) to ensure long-term safety of stochastic systems under uncertain dynamics and evolving human specifications. We proved that the proposed safety condition guarantees non-decreasing safety probability over time and can be seamlessly integrated with optimization-based control. We further integrate multi-turn large language models (LLMs) leveraging past data and interactions to account for varying human preferences over time. Through autonomous lane-keeping experiments, we demonstrated that the proposed method maintains consistent long-term safety across varying road conditions, adapts effectively to human preferences, and achieves improved safety–efficiency trade-offs. Future work includes extensions to reinforcement learning and multi-agent settings, with potential adversarial attacks~\citep{hu2025steering}.

\acks{This work is sponsored in part by the National Science Foundation CAREER program (no. 2442948), in part by a grant from Japan Science and Technology Agency, in part by a grant from the Commonwealth of Pennsylvania, Department of Community and Economic Development, and in part by Grant-in-Aid for Scientific Research (KAKENHI) from the Japan Society for Promotion of Science (23K13354, 26K07557). We thank Lordphone Wen for contributions to the ablation experiments in the camera-ready version of this paper.}

\bibliography{citation}

\appendix

\section{Proofs}
\label{sec:appendix_proof}

\begin{proof} (Theorem~\ref{thm:main_theorem_aug})
We use mathematical induction to prove \eqref{eq:satisfy_control_policy_aug}. 
Condition \eqref{eq:satisfy_control_policy_aug} holds for $k=0$ due to the assumption on initial condition. Next, we suppose \eqref{eq:satisfy_control_policy_aug} holds at time $k\geq0$, and show \eqref{eq:satisfy_control_policy_aug} holds for time $k+1$. Let
\begin{align}
\label{eq:greater_than_q}
    \mE_{\pi_\mathrm{safe}}\left[\sprob^\pi_{H_k}(X_k)  \right]= 1 - \epsilon + q
\end{align}
for some $q\geq 0$. Note that the expectation in~\eqref{eq:greater_than_q} is taken over all possible trajectories starting from $X_0$. For the rest of the proof, we omit $\pi_\mathrm{safe}$ in the expectation for notation conciseness.
We first consider the set of events $\mathcal{D}_0 $ that satisfies $\sprob^\pi_{H_{k}}(\state_k) < 1-\epsilon$ and $\mathcal{D}_1$ that satisfies $\sprob^\pi_{H_{k}}(\state_k) \geq 1-\epsilon $. 
Then, define variables $v_i$ and $\delta_i$, for $i\in\{0,1\}$, as follows:
\begin{align}
    \label{eq:V0}
    v_0 & = \mE \left[\sprob^\pi_{H_{k}}(\state_k) \mid \mathcal{D}_0, \pi_\mathrm{safe} \right] = 1-\epsilon-\delta_0, \\
    \label{eq:V1}
    v_1 & = \mE \left[\sprob^\pi_{H_{k}}(\state_k) \mid \mathcal{D}_1, \pi_\mathrm{safe} \right] = 1-\epsilon+\delta_1, 
\end{align}
where $\delta_0 \geq 0$ and $\delta_1 \geq 0$, since $\mE \left[\sprob^\pi_{H_{k}}(\state_k) \mid \mathcal{D}_0 \right] < 1-\epsilon$ and $
\mE \left[\sprob^\pi_{H_{k}}(\state_k) \mid \mathcal{D}_1\right] \geq 1-\epsilon$ hold. 
The left hand side of \eqref{eq:greater_than_q} can then be written as
\begin{equation}
\begin{aligned}
    \mE[\sprob^\pi_{H_{k}}(\state_k)]
    & = \mE \left[\sprob^\pi_{H_{k}}(\state_k) \mid \mathcal{D}_0\right] \mP(\mathcal{D}_0) + \mE \left[\sprob^\pi_{H_{k}}(\state_k) \mid \mathcal{D}_1 \right] \mP(\mathcal{D}_1) \\
    & = v_0q_0 + v_1q_1, 
\label{eq:expectation_in_VP}
\end{aligned}
\end{equation}
where $q_i := \mP(\mathcal{D}_i)$, for $i \in \{0, 1\}$, and we have
\begin{align}
\label{eq:fact2}
    \mP(\mathcal{D}_0) + \mP(\mathcal{D}_1) = q_0 + q_1 = 1.
\end{align}
Combining \eqref{eq:greater_than_q} and \eqref{eq:expectation_in_VP}, and applying~\eqref{eq:V0} and~\eqref{eq:V1} gives
\begin{align}
\label{eq:combine_expectation}
    1-\epsilon+q & = v_0 q_0 + v_1 q_1  \\
                 & =\left(1-\epsilon-\delta_0 \right)q_0 + \left(1-\epsilon+\delta_1 \right)q_1. 
\end{align}
This combined with \eqref{eq:fact2} yields
\begin{align}
\label{eq:fact3}
    q=\delta_1 q_1 - \delta_0 q_0.
\end{align}
On the other hand, we have
\begin{align}
    & \mE \left[\gamma\left(\sprob^\pi_{H_{k}}(\state_k) - (1-\epsilon) \right) \right] \label{eq:expectation_derivation_start} \\
    = \; & \mP(\mathcal{D}_0)  \mE \left[\gamma\left(\sprob^\pi_{H_{k}}(\state_k) - (1-\epsilon) \right) \mid \mathcal{D}_0 \right]   + \mP(\mathcal{D}_1)  \mE \left[\gamma\left(\sprob^\pi_{H_{k}}(\state_k) - (1-\epsilon) \right) \mid \mathcal{D}_1 \right] \\
    =  \; & q_0 \left( \mE \left[\gamma\left(\sprob^\pi_{H_{k}}(\state_k) - (1-\epsilon) \right) \mid \mathcal{D}_0 \right] \right)  + q_1 \left( \mE \left[\gamma\left(\sprob^\pi_{H_{k}}(\state_k) - (1-\epsilon) \right) \mid \mathcal{D}_1 \right] \right) \label{eq:expectation_given_q}  \\
    \leq \; & q_0 \left( \gamma\left(\mE \left[\sprob^\pi_{H_{k}}(\state_k) - (1-\epsilon) \mid \mathcal{D}_0 \right] \right) \right) + q_1 \left( \gamma\left( \mE \left[\sprob^\pi_{H_{k}}(\state_k) - (1-\epsilon) \mid \mathcal{D}_1 \right] \right) \right) \label{eq:inequality_from_jensen_rule} \\
    = \; & q_0 \left( \gamma\left(-\delta_0 \right) \right) + q_1 \left( \gamma\left(\delta_1 \right) \right) \label{eq:jense_inequality_given_V0V1} \\
    \leq \; & \gamma \left(-q_0\delta_0 + q_1\delta_1 \right) \label{eq:jensen_inequality_given_V0V1_assume_A2}  \\
    = \; & \gamma(q) \label{eq:expectation_lessthan_0}  \\
    \leq \; & q. \label{eq:expectation_derivation_end} 
\end{align}
Here, \eqref{eq:inequality_from_jensen_rule} is obtained from Jensen's inequality for concave function $\gamma$; \eqref{eq:jense_inequality_given_V0V1} is based on \eqref{eq:V0} and \eqref{eq:V1}; \eqref{eq:jensen_inequality_given_V0V1_assume_A2} is due to Jensen's inequality, design requirement 1 and \eqref{eq:fact2}; and \eqref{eq:expectation_lessthan_0} is due to \eqref{eq:fact3}. From \eqref{eq:expectation_derivation_start} to \eqref{eq:expectation_derivation_end}, we have
\begin{align}
\label{eq:less_than_minus_q}
    \mE \left[-\gamma\left(\sprob^\pi_{H_{k}}(\state_k) - (1-\epsilon) \right) \right]\geq -q.
\end{align}
Recall that the control action is chosen to satisfy \eqref{eq:safety_condition_each_Zt_aug}. 
Now, we take the expectation over both side of \eqref{eq:safety_condition_each_Zt_aug} to obtain
\begin{align}
\label{eq:expectation_over_condition}
    &\mathbb{E}[\mathcal{A}^{U_k}_{H_{k+1}} \Psi^\pi_{H_k} (\state_k) ]   
    \geq \mathbb{E}[ -\gamma(\sprob^\pi_{H_{k}}(\state_k)-(1-\epsilon))].
\end{align}
From the definition \eqref{eq:generator_A} of the discrete-time generator, we have
\begin{equation}
\begin{aligned}
    & \mathbb{E}\left[ \dfrac{\mathbb{E}[ \Psi^\pi_{H_{k+1}}(\state_{k+1})  \mid \state_k, U_k] - \Psi^\pi_{H_{k}}(\state_k) }{\Delta t}  \right]
    \geq
    \mathbb{E}[-\gamma(\sprob^\pi_{H_{k}}(\state_k)-(1-\epsilon)) ].
\end{aligned}
\end{equation}
Using the law of total expectation, we have
\begin{align}
\label{eq:after_law_of_total_expectation}
    & \frac{\mathbb{E}[\sprob^\pi_{H_{k+1}}(\state_{k+1})-\sprob^\pi_{H_{k}}(\state_k) ]}{\Delta t}
    \geq 
    \mathbb{E}[-\gamma(\sprob^\pi_{H_{k}}(\state_k)-(1-\epsilon)) ].
\end{align}
Combining \eqref{eq:greater_than_q}, \eqref{eq:less_than_minus_q}, \eqref{eq:after_law_of_total_expectation} and design requirement 2 yields
\begin{align}
\label{eq:eps_q_deltat_c_start}
    & \mathbb{E}[\sprob^\pi_{H_{k+1}}(\state_{k+1})  ] \\
    \geq \;& \mathbb{E}[\sprob^\pi_{H_{k}}(\state_k) ]  
    + \mathbb{E}[-\gamma(\sprob^\pi_{H_{k}}(\state_k)-(1-\epsilon)) ]\Delta t\\
    \geq \;& 1-\epsilon+q-q\Delta t\\
    = \; &1-\epsilon+q(1-\Delta t).
\end{align}
Since $\Delta t\ll 1$ and $q\geq 0$, we have
\begin{align}
\label{eq:safety_condition_proof}
    \mathbb{E}[\sprob^\pi_{H_{k+1}}(\state_{k+1})] \geq 1-\epsilon.
\end{align}
Here, the left hand side of~\eqref{eq:safety_condition_proof} is equivalent to the left hand side of~\eqref{eq:satisfy_control_policy_aug}.
\end{proof}

\section{Implementation of Safety Constraints}
\label{sec:appendix_psc_calculation}

In this section, we discuss a numerical method for calculating the discrete time generator $\mathcal{A}^{U_k}_{H_{k+1}}$
$\Psi^{\pi}_{H_k} (X_k)$ in~\eqref{eq:generator_A}. 
Given $H_{k+1}$ and $H_k$, the value of the discrete time generator can be calculated as 
\begin{align}
 & \mathcal{A}^{U_k}_{H_{k+1}} \Psi^{\pi}_{H_k} (X_k) = 
 \dfrac{ \mathbb{E}[\Psi^\pi_{H_{k+1}}(X_{k+1}) | X_k, U_k] - \Psi^\pi_{H_k}(X_k) }{\Delta t} \notag \\
 = \; & 
 \underbrace{ 
 \dfrac{ \mathbb{E}[\Psi^\pi_{H_{k+1}}(X_{k+1}) | X_k, U_k] - \Psi^\pi_{H_{k+1}}(X_k) }{\Delta t} 
 }_{ := \mathcal{S}^{U_k}\Psi^\pi_{H_{k+1}}(X_k) } + \underbrace{ \dfrac{ \Psi^\pi_{H_{k+1}}(X_k) - \Psi^\pi_{H_k}(X_k) }{\Delta t}
 }_{ := \mathcal{T}_{H_{k+1}}\Psi^\pi_{H_k}(X_k)}
 , 
  \label{eq:generator_decomposition}
\end{align}
where on the right hand side the first term $\mathcal{S}^{U_k} \Psi^\pi_{H_{k+1}}(X_{k})$ involves prediction of the state $X_{k+1}$ at the next time step, and the second term $\mathcal{T}_{H_{k+1}}\Psi^\pi_{H_k}(X_k)$ involves the update of the density function from $H_k$ to $H_{k+1}$. 
Here, we require $H_{k+1}$ at the time step $k$ to calculate these terms. 
The function $H_k$ is supposed to be constructed from the information obtained up to the time step $k-1$. 
Then we know $H_{k+1}$ is constructed at the time step $k$, and we can directly evaluate the second term $\mathcal{T}_{H_{k+1}}\Psi^\pi_{H_k}(X_k)$ when we have estimates for $\Psi^\pi_{H_{k+1}}(X_k)$ and $\Psi^\pi_{H_k}(X_k)$.  
Furthermore, to approximate the first term $\mathcal{S}^{U_k} \Psi^\pi_{H_{k+1}}(X_{k})$, we suppose the dynamics \eqref{eq:discrete_dynamics} are given as a discretization of continuous-time dynamics 
\begin{align}
\label{eq:continuous_dynamics}
     d X_t &= f_c(X_t, U_t, \xi) dt + \sigma_c(X_t,U_t) dW_t, \; \forall t\in\R^+,
\end{align}
where $X_0 = x_0$ and we denote the state by $X_t$ and the input by $U_t$ with slight abuse of notation.
We then have the following theorem. 

\begin{assumption}
\label{asp:strong_sol_SDE}
    The continuous dynamics function $f_c$ and $\sigma_c$ in~\eqref{eq:continuous_dynamics} are twice differentiable with respect to $X_t$. The control space $\mathcal{U}$ is compact.
\end{assumption}

\begin{assumption}
\label{asp:differentiable_safe_prob}
    The safety probability $\Psi_{H}^{\pi}(X)$ is differentiable with regard to $X$ for any $H$ and $\pi$, and has bounded derivatives.
\end{assumption}

\begin{theorem}
\label{thm:infinitesimal_calculation}
Suppose that Assumption~\ref{asp:strong_sol_SDE} and~\ref{asp:differentiable_safe_prob} hold.
Then, the following equality holds: 
\begin{align}
&\lim_{\Delta t \to 0} \mathcal{S}^{U_k} \Psi^\pi_{H_{k+1}}(X_{k}) 
 \notag \\ 
 = \; &  \nabla_x \Psi^\pi_{H_{k+1}}(\state_k) \mathbb{E} \left[ f_c (\state_k, U_k, \Xi_{k+1}) \mid X_k, U_k \right] + \frac{1}{2}\sigma_c(X_k,U_k)^\top \sigma_c(X_k,U_k)\nabla^2_x \Psi^\pi_{H_{k+1}}(\state_k),
\end{align}
where $\nabla_x$ and $\nabla^2_x$ are the gradient and Hessian operators. 
\end{theorem}


\begin{proof} 
(Theorem~\ref{thm:infinitesimal_calculation})
From the definition of $\mathcal{S}^{U_k} \Psi^\pi_{H_{k+1}} (\state_k)$ in~\eqref{eq:generator_decomposition} we get
\begin{align}
     & \lim_{\Delta t \to 0} \mathcal{S}^{U_k} \Psi^\pi_{H_{k+1}} (\state_k) \\
    = & \lim_{\Delta t \to 0} \dfrac{ \mathbb{E}[\Psi^\pi_{H_{k+1}}(X_{k+1}) | X_k, U_k] - \Psi^\pi_{H_{k+1}}(X_k) }{\Delta t} 
   \\ 
      = & \lim_{\Delta t \to 0}
    \mE \left[ \dfrac{\Psi^\pi_{H_{k+1}}(\state_{k+1})  - \Psi^\pi_{H_{k+1}}(\state_k) }{\Delta t} \middle| X_k, U_k\right] 
    \\ 
    = & \lim_{\Delta t \to 0} \mathbb{E} \Bigg[  \mathbb{E} \Bigg[ \dfrac{\Psi^\pi_{H_{k+1}}(\state_{k+1})  - \Psi^\pi_{H_{k+1}}(\state_k) }{\Delta t} \bigg| \state_k, U_k, \Xi_{k+1}\Bigg] \bigg| X_k, U_k \Bigg] \label{eq:separate_expectation} 
    \\ 
    = & \mathbb{E} \Bigg[ \lim_{\Delta t \to 0}  \mathbb{E} \Bigg[ \dfrac{\Psi^\pi_{H_{k+1}}(\state_{k+1})  - \Psi^\pi_{H_{k+1}}(\state_k) }{\Delta t}  \bigg| \state_k, U_k, \Xi_{k+1}\Bigg] \bigg| X_k, U_k \Bigg],
    \label{eq:move_expectation}
\end{align}
where \eqref{eq:separate_expectation} holds due to the law of total expectation, and \eqref{eq:move_expectation} holds due to Assumption~\ref{asp:differentiable_safe_prob} and the dominated convergence theorem~\citep{flanders1973differentiation}. 
Then we have
\begin{align}
    & \lim_{\Delta t \to 0} \mathcal{S}^{U_k} \Psi^\pi_{H_{k+1}} (\state_k) \\
     =\;& \mathbb{E} \Bigg[ \lim_{\Delta t \to 0}  \mathbb{E} \Bigg[ \dfrac{\Psi^\pi_{H_{k+1}}(\state_{k+1})  - \Psi^\pi_{H_{k+1}}(\state_k) }{\Delta t}  \mid \state_k, U_k,\Xi_{k+1} \Bigg] \bigg| X_k, U_k \Bigg]
     \\ 
     =\;& \mathbb{E} \bigg[ \nabla_x \Psi^\pi_{H_{k+1}}(\state_k) f_\mathrm{c} (\state_k, U_k,  \Xi_{k+1})
     + \frac{1}{2}\sigma_c(X_k, U_k)^\top \sigma_c(X_k, U_k) \nabla^2_x \Psi^\pi_{H_{k+1}}(\state_k) \bigg| X_k, U_k \bigg] \label{eq:Ito_lemma}
     \\ 
     =\;& \nabla_x \Psi^\pi_{H_{k+1}}(\state_k) \mathbb{E}\left[ f_c (\state_k, U_k, \Xi_{k+1}) \mid X_k, U_k \right] + \frac{1}{2}\sigma_c(X_k,U_k)^\top \sigma_c(X_k,U_k)\nabla^2_x \Psi^\pi_{H_{k+1}}(\state_k),
\end{align}
where~\eqref{eq:Ito_lemma} holds by applying Ito's lemma under Assumption~\ref{asp:strong_sol_SDE} and~\ref{asp:differentiable_safe_prob}.
\end{proof}

\begin{remark}
Evaluation of the discrete time generator $\mathcal{A}^{U_k}_{H_{k+1}} \Psi^{\pi}_{H_k} (X_k)$ in~\eqref{eq:generator_A} requires values of the long-term safety probability $\Psi^{\pi}_{H_k}(X_k)$ in~\eqref{eq:safety_prob}. In practice, one could run online parallel Monte Carlo simulations to empirically estimate such values, or leverage offline model-based and data-driven methods such as the ones in~\citet{chern2021safe, wang2025generalizable}.
\end{remark}

\section{Adaptive PSC Algorithm}
\label{sec:appendix_psc_alg}

In this section, we present the overall adaptive probabilistic safety certificate algorithm in Algorithm~\ref{alg:safe_control}. 
In line~\ref{ln:initialization}, we initialize the horizon of the problem $T_{\text{end}}$, the discrete time step $\Delta t$, the initial state $X_0$, the nominal control policy $\pi$, the long-term safety horizon $T$, the risk tolerance $\epsilon$, and the reference controller $N$. While time step $k$ has not reached $T_{\text{end}}$, we run the proposed adaptive safe control method. Specifically, we first obtain the available information $H_k$ from the estimator (line~\ref{ln:estimator}), calculate the safety probability $\Psi^\pi_{H_{k}}(X_{k})$ (line~\ref{ln:safe_prob}), solve optimization~\eqref{eq:conditioning} to get safe control $U_k^{\text{safe}}$ (line~\ref{ln:safe_control}), and step the dynamics to get state at the next time step $\state_{k+1}$ (line~\ref{ln:step_dynamics}). Note that the safety constraint for the optimization problem~\eqref{eq:conditioning} involves the calculation of the the infinitesimal generator $\mathcal{A}^{U_k}_{H_{k+1}} \Psi^{\pi}_{H_k} (X_k)$. In practice one can assume finite admissible control set for efficient computation of the problem.

\begin{algorithm}[H]
\caption{Adaptive Probabilistic Safety Certificate}\label{alg:safe_control}
\LinesNumbered
\SetAlgoLined

\KwIn{$T_{\text{end}}, \Delta t, X_0, \pi, T, \epsilon, N$}

$k \gets 0$\; \label{ln:initialization}

\While{$k < T_{\text{end}}$}{
    Obtain $H_k$ from the estimator\; \label{ln:estimator}
    
    Calculate safety probability $\Psi^\pi_{H_k}(X_k)$\; \label{ln:safe_prob}
    
    Solve~\eqref{eq:conditioning} to get safe control
    $U_k^{\text{safe}} \gets \pi_{\text{safe}}(X_k, \hat{\xi}_k)$\; \label{ln:safe_control}
    
    Step the dynamics~\eqref{eq:discrete_dynamics} with $U_k^{\text{safe}}$
    to get $X_{k+1}$\; \label{ln:step_dynamics}
    
    $k \gets k + 1$\;
}
\end{algorithm}

\section{Vehicle model}
\label{sec:appendix_vehicle_model}

In this section, we introduce the vehicle model used in numerical experiments. 
For the vehicle dynamics, a 3 degree-of-freedom (DoF) model is adopted, and the 
dynamics of the state $x_\mathrm{vehicle} := [v_x, v_y, r, \delta]$,  where $v_x$ and $v_y$ are the vehicle longitudinal and lateral velocities, respectively, $r$ is the yaw rate, and $\delta$ is the steering angle, is given by the following equations: 
\begin{align}
    m \dot{v}_x = & m v_y r 
      +(F_\text{Lfl}+ F_\text{Lfr})\cos\delta 
     \notag \\
      & -(F_\text{Sfl}+ F_\text{Sfr})\sin\delta 
      + F_\text{Lrl} + F_\text{Lrr}\\
    m \dot{v}_y = & -m v_y r 
      +(F_\text{Sfl}+ F_\text{Sfr})\cos\delta 
      \notag \\
      & +(F_\text{Lfl}+ F_\text{Lfr})\sin\delta 
      +  F_\text{Srl} + F_\text{Srr}\\
    I_z \dot{r}  = & 
      l_f \left\{ (F_\text{Sfl}+F_\text{Sfr})\cos\delta + (F_\text{Lfl}+F_\text{Lfr})\sin\delta\right\} 
      \notag \\
    & - l_r \left(F_\text{Srl} + F_\text{Srr} \right) 
     + \frac{W}{2} \left(-F_\text{Lrl} + F_\text{Lrr}\right) 
    \nonumber \\
    & + \frac{W}{2} \left\{ (-F_\text{Lfl}+F_\text{Lfr})\cos\delta +(F_\text{Sfl}-F_\text{Sfr})\sin\delta  \right\}
    \\
    \dot{\delta} = & \Delta \delta
\end{align}
where the parameters $m$ is the mass of the vehicle,  $I_z$ is the rotational inertia of the vehicle, $l_f$ is the front wheel distance to vehicle center, $l_r$ is the rear wheel distance to vehicle center, and $W$ is the width of the vehicle. The symbol $F_{ij}$ for $i \in \{\mathrm{L}, \mathrm{S}\}$ and $j \in \{ \mathrm{fl}, \,\mathrm{fr}, \,\mathrm{rl}, \,\mathrm{rr} \}$ stands for the longitudinal and side (lateral) tire forces for each of the four tires (fl:front left, fr:front right, rl:rear left, rr:rear right). 
Each tire force can be characterized as a function of the tire slip angle $\alpha_{ij}$ and slip ratio $\lambda_{ij}$ as follows 
\begin{align}
  \alpha_\mathrm{fl} = & \,\alpha_\mathrm{fr} = \delta - \dfrac{ v_y + l_\mathrm{f} r }{v_x}, \\
  \alpha_\mathrm{rl} = & \, \alpha_\mathrm{rr} =  - \dfrac{ v_y + l_\mathrm{r} r }{v_x},
  \\ 
  \lambda_{ij} = & \dfrac{R_\mathrm{e} \omega_{ij} - v_x }{ \max(R_\mathrm{e}\omega_{ij}, v_x ) } 
\end{align}
where $R_\mathrm{e}$ is the wheel radius. 
As the tire slip angle and the slip ratio increase, the tire forces saturate rapidly and transit from linear to nonlinear dynamics. 
To represent the nonlinear tire forces, following \cite{li2023}, we use the LuGre combined-slip tire model, which considers the internal tire characteristics and has been widely utilized to describe the nonlinearity under extreme conditions. 
The tire longitudinal and lateral force descriptions are as follows:
\begin{align}
    F_{\mathrm{L}ij} = & \left( \dfrac{\sigma_{0x}}{\frac{\sigma_{0x} \| v_{r, ij}\|}{ \mu g(v_{r, ij})  } + \kappa_x R_\mathrm{e} |\omega_{ij}| } + \sigma_{2x} \right) v_{rx, ij} F_z, \\
    F_{\mathrm{S}ij} = & \left( \dfrac{\sigma_{0y}}{\frac{\sigma_{0y} \| v_{r,ij}\|}{ \mu g(v_{r,ij})  } + \kappa_y R_\mathrm{e} |\omega_{ij}| } + \sigma_{2y} \right) v_{ry, ij} F_z, 
\end{align}
where $v_{rx,ij}$, $v_{ry,ij}$ and $\| v_{r,ij} \|$ are given by 
\begin{align}
 v_{rx,ij} =& R_\mathrm{e} \omega_{ij} - v_x, \\
 v_{ry, ij} =& v_x \alpha_{ij}, \\
 \| v_{r,ij} \| = & \sqrt{ v_{rx,ij}^2 + v_{ry, ij}^2}. 
\end{align}
In the above equations, $\mu$ stands for the road  friction coefficient, and $g(v_{r,ij})$ is a Stribeck function given by
\begin{align}
    g(v_{r,ij}) = \mu_\mathrm{c} + (\mu_\mathrm{s} - \mu_\mathrm{c} ) \exp \left(- \sqrt{ v_{r,ij}/V_s 
    } \right),
\end{align}
where $V_s$ is the Stribeck relative velocity. 
For the dynamics of the state $x_\mathrm{wheel}$, 
we assume an open drivetrain with equal torque distribution and linear dynamics, as shown below:
\begin{align}
    I_\omega \dot{\omega}_\text{fl} =& -R_\mathrm{e} F_\text{Lfl} + \frac{1}{4}\tau_\mathrm{e}, \\
    I_\omega \dot{\omega}_\text{fr} =& -R_\mathrm{e} F_\text{Lfr} + \frac{1}{4}\tau_\mathrm{e}, \\
    I_\omega \dot{\omega}_\text{rl} =& -R_\mathrm{e} F_\text{Lrl} + \frac{1}{4}\tau_\mathrm{e}, \\
    I_\omega \dot{\omega}_\text{rr} =& -R_\mathrm{e} F_\text{Lrr} + \frac{1}{4}\tau_\mathrm{e}, \\
    \dot{\tau}_\mathrm{e} = & \Delta \tau_\mathrm{e} 
\end{align}
where $J_\omega$ is the wheel inertia. 
To formulate a lane keeping problem, the vehicle dynamics model needs to be combined with the model of the vehicle position and heading angle. 
When they are described in a road coordinate, the road curvature can be viewed as a disturbance pushing the vehicle out of the lane. 
By defining the yaw rate error $\psi$ between the vehicle and the road, its dynamics can be described as 
\begin{align}
    \dfrac{\mathrm{d}\psi}{\mathrm{d} t}  & = r - v_x \rho(s)  
\end{align}
where $\rho$ is the curvature of the road and can be regarded as a function of the distance $s$ along the road.  
The dynamics of the distance $s$ and the lateral error $e$ can be obtained by solving the following equations:
\begin{align}
     \dfrac{\mathrm{d} s}{\mathrm{d} t} &=  v_x \cos\psi - v_y \sin\psi, \\
    \dfrac{\mathrm{d} e }{\mathrm{d} t} &=  v_y \cos\psi + v_x \sin\psi. 
\end{align} 
The list of parameters and their values are summarized in Table~\ref{tab:parameters}.

\begin{table}[t]
\caption{Parameter Setting} 
 \centering
\begin{tabular}{c l c}
\hline
\textbf{Symbol} & \textbf{Definition} & \textbf{Value} \\
\hline
$m$ & Vehicle mass & $\SI{1430}{kg}$ \\
$R_\mathrm{e}$ & Wheel radius & $\SI{0.325}{m}$\\ 
$I_z$ & Yaw moment of inertia & $\SI{2059}{kg .m^2}$\\
$I_\omega$ & Wheel moment of inertia & $\SI{1.68}{kg.m^2}$ \\
$L_f$ & Distance from CG to front axis & $1.05 \mathrm{~m}$ \\
$L_r$ & Distance from CG to rear axis  & $1.61 \mathrm{~m}$ \\
$W$ & width of the vehicle & $\SI{1.55}{m}$\\
$V_\mathrm{s}$ & Stribeck relative velocity & $\SI{6.6}{m/s}$ \\
$\sigma_{0x}$ & Longitudinal rubber stiffness & $\SI{195}{m^{-1}}$ \\
$\sigma{2x}$ & Longitudinal relative viscous & $\SI{0.001}{s/m}$ \\
$\kappa_x$ & Longitudinal load distribution factor & 13.4 \\
$\sigma_{0y}$ & Lateral rubber stiffness & $\SI{195}{m^{-1}}$ \\ 
$\sigma_{2y}$ & Lateral relative viscous damping & $\SI{0.001}{s/m}$ \\
$\kappa_y$ & Lateral load distribution factor & 13.4 \\
$\mu_\mathrm{s}$ & Static friction coefficient & 0.55 \\
$\mu_\mathrm{c}$ & Dynamics friction coefficient & 0.35 \\
$g$ & gravitational acceleration & $9.8~\textrm{~m/s}^2$ \\
\hline
\end{tabular}
\label{tab:parameters}
\end{table}

\section{Experiment Details}
\label{sec:appendix_exp_detail}

In this section, we provide experiment details corresponding to Section~\ref{sec:experiment}. 

\subsection{General Setup}
\label{sec:sub_appendix_general_setup}

For the state space in~\eqref{eq:vehicle_state}, when the vehicle is traveling on a road with a non-zero curvature, the curvature is viewed as a disturbance on the heading error $\psi$ described by
\begin{align}
    \dot{\psi} = r - v_x \rho(s),
\end{align}
where $\rho(s)$ denotes the radius of curvature as a function of $s$. 

The road-tire friction coefficient $\mu$ is an unknown fixed parameter. It typically takes a value from 0.7 to 0.9 for a dry road condition, from 0.4 to 0.7 for a wet road condition, and from 0.2 to 0.4 for an icy road condition \citep{rajamani2012}.  
While various methods including extended Kalman filtering~\citep{Ray1997,Enisz2015} are proposed for online estimation of the friction coefficient, in all experiments we chose to use a Bayesian estimator.
Note that for the estimator~\eqref{eq:estimator_update} Theorem \ref{thm:main_theorem_aug} holds without assuming Gaussian distributions for the parameter estimates.

For all three MPC-based adaptive safe control methods considered, we define the MPC cost function as  
\begin{align}
    J_{\text{MPC}} = \sum_{k=1}^{T_{\mathrm{mpc}}} 0.05 (v_x(k)-V_\mathrm{ref})^2 + e(k)^2 + \psi(k)^2,
\end{align}
where $T_{\mathrm{mpc}}$ is the MPC prediction horizon. 
We implement the nonlinear MPC controller using \texttt{nlmpc} provided by the Model Predictive Control Toolbox in MATLAB with automatic C code generation. 
We set the sampling time to \SI{0.2}{s}, the control horizon to 2 time steps. The initial velocity of the vehicle is chosen as $v_0 = 20 \mathrm{km/h}$.
For the \textbf{AMPC} method, we use the expectation $\mu_k$ of the posterior $P_k(\mu|H_k)$ in the prediction model for adaptation. 
For control-dependent barrier functions (\textbf{CDBF})~\citep{Huang2021,YW2022}, it constrains the vehicle states within a time-varying and control-dependent lateral stability region, and such constraints are incorporated inside the nonlinear adaptive MPC.
For the proposed method, we integrate the adaptive probabilistic safety certificate (PSC)~\eqref{eq:conditioning} with the MPC framework by adding the safety constraint \eqref{eq:safety_condition_each_Zt_aug}.

\begin{figure}[t!]
    \centering
    \begin{minipage}[b]{0.45\linewidth}
        \centering
        \includegraphics[width=\linewidth]{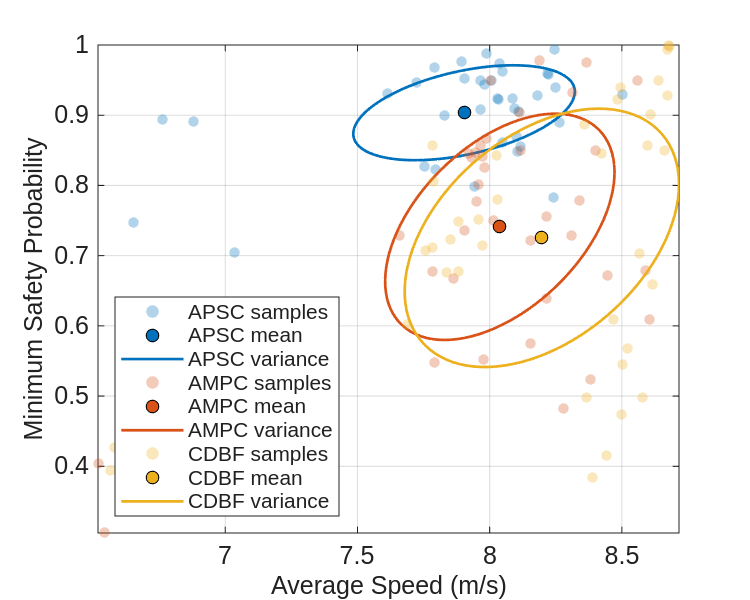}
        \caption*{(a) Feasible cases. }
        \label{fig:tradeoff_feasible}
    \end{minipage}
    \begin{minipage}[b]{0.45\linewidth}
        \centering
        \includegraphics[width=\linewidth]{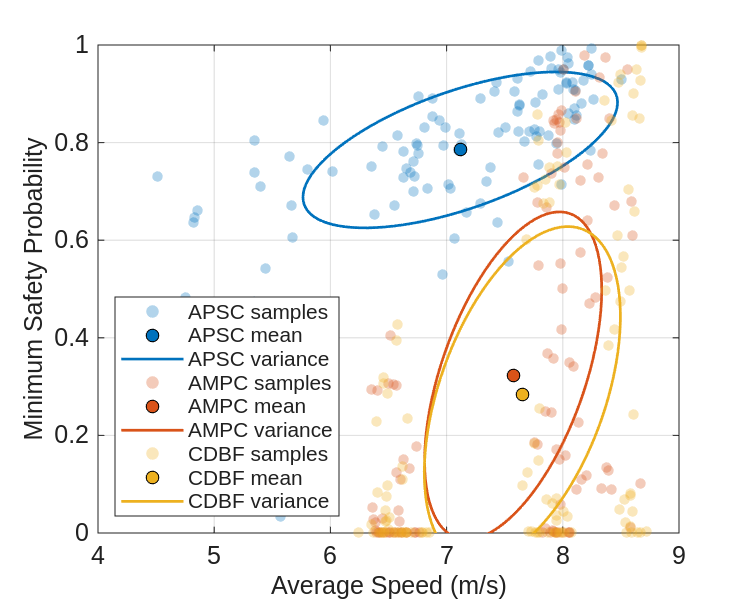}
        \caption*{(b) All cases.}
        \label{fig:tradeoff_all}
    \end{minipage}
    \caption{Safety vs. efficiency trade-offs with MPC-based adaptive safe control variants. \textbf{APSC:} Proposed adaptive probabilistic safety certificate. \textbf{AMPC:} Adaptive model predictive control. \textbf{CDBF:} Control-dependent barrier functions.}
    \label{fig:tradeoff_full}
\end{figure}

\subsection{Adaptive Safe Control}
\label{sec:sub_appendix_adaptive_psc}

For the experiments in Section~\ref{sec:adaptive_mpc_exp}, the online computation time at each control step is reported in Fig.~\ref{fig:horizon_vs_computation}.
All computation time is measured on a workstation with AMD EPYC 7763@2.45GHz. 
As shown, the computation time of all methods increases with the MPC horizon. Combined with the safety results in Fig.~\ref{fig:horizon_vs_performance}, these findings indicate that only the proposed method achieves the desired long-term safety within real-time computation using short MPC horizons.

\begin{figure}
 \centering
\includegraphics[width=0.5\linewidth]{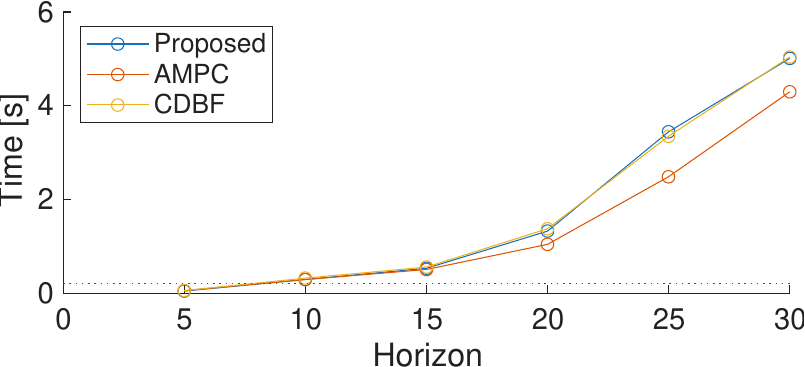} 
 \caption{Per step online computation time vs. MPC horizon $T_{\mathrm{mpc}}$.} 
 \label{fig:horizon_vs_computation}
\end{figure}

Among the 108 distinctive cases for quantitative analysis in Section~\ref{sec:adaptive_mpc_exp}, settings that yield a safety probability below 0.3 for all methods are marked as infeasible, as the initial conditions and safety requirements are too stringent for any method to satisfy (the sufficiently safe initial condition in Theorem~\ref{thm:main_theorem_aug} is violated). Among the 108 configurations, 37 are feasible. Fig.~\ref{fig:tradeoff_full} visualizes the trade-offs for both the feasible and overall cases. We can see that the proposed adaptive PSC consistently achieves higher safety probability than that of AMPC and CDBF, while maintaining comparable average vehicle velocities. These results show that the proposed method achieves better trade-offs between safety and efficiency by ensuring long-term safety without compromising performance.

\subsection{LLM Integration}
\label{sec:sub_appendix_llm_setup}


We provide full quantitative results for LLM integration with adaptive PSC. 
For case (i), where the human exhibits varying driving aggressiveness preferences over time, we report the results in Table~\ref{tab:llm_case_1_full_icy}. It can be seen that across diverse LLM models with different sizes and capabilities, the proposed framework can adapt to human preferences accordingly, while consistently ensuring safety of the system under road condition uncertainties.
For case (ii), where the human provides erroneous judgments about road conditions, we provide the full quantitative results in Table~\ref{tab:llm_case_2_full}. It can be observed that across diverse LLM models, the proposed framework can leverage historical data to refine the prior estimate of the friction coefficient, and therefore achieves smoother yet safe and efficient subsequent runs despite persistent erroneous human instruction. 

The example human instructions for different categories used in the experiments are shown in Table~\ref{table:user_inputs}. The empirical \textit{Safety} reported in the results are calculated through the ratio of time period where the lateral lane deviation is less than $3\mathrm{m}$

{\tiny
\begin{longtable}[c]{@{}lll@{}}
\toprule
\textbf{Aggressive User Input} & \textbf{Conservative User Input} & \textbf{Dry and Unsure User Input} \\ 
\midrule
\endfirsthead

\toprule
\multicolumn{2}{l}{\textit{Continuation of Table~\ref{table:user_inputs}}} \\ 
\midrule
\textbf{Aggressive User Input} & \textbf{Conservative User Input} & \textbf{Dry and Unsure User Input}\\ 
\midrule
\endhead

\midrule
\multicolumn{2}{r}{\textit{Continued on next page}} \\ 
\bottomrule
\endfoot

\bottomrule
\endlastfoot

I want to drive aggressively and push the limits. &	That was too fast, please slow down. &	The road seems dry, but I'm not entirely sure. \\
Go for maximum performance and speed. &	I want a more conservative setup. &	It looks like the road is dry, though I can't be 100\% certain. \\
Make it react faster and turn sharper. &	Reduce speed, it felt unsafe. &	The road appears to be dry, but I could be mistaken. \\
I prefer bold and assertive driving. &	Too aggressive—make it smoother. &	Seems like a dry surface, but not completely confident. \\
Drive with high energy and confidence. &	Please prioritize safety over speed. &	I think the road is dry, though there's some doubt. \\
Prioritize speed and agility over comfort. &	I prefer a slower and steadier control. &	The surface looks dry, but I'm not fully sure. \\
I want a powerful, race-style control. &	Reduce acceleration, it's too strong. &	The road feels dry, but I can't say for sure. \\
Let it take corners fast and hard. &	Too risky on this road. &	I believe the road is dry, though not completely sure. \\
Push acceleration and minimize hesitation. &	I want gentler lane corrections. &	The pavement seems dry, but I might be wrong. \\
I want fast, strong, and responsive steering. &	Too harsh—make it more stable. &	Probably dry road conditions, but not confirmed. \\
Use the most dynamic, performance-oriented mode. &	It felt unsafe, drive more cautiously. &	Looks like dry asphalt, though uncertain. \\
Allow wider lane deviation for faster maneuvers. &	Please slow the response a bit. &	I think it's dry, but can't be 100\% confident. \\
Give me riskier but more exciting control. &	It's too aggressive for my taste. &	The road seems fine and dry, but I could be wrong. \\
I want a high-speed, competitive setup. &	Too fast, not comfortable. &	Dry conditions, I think, but not fully confident. \\
Drive aggressively even when conditions vary. &	Please reduce throttle aggressiveness. &	It appears dry, but I'm not completely convinced. \\
Maximize performance, not stability. &	I want smoother driving this time. &	Road looks dry overall, though not totally sure. \\
I want it to move quickly, even if precision drops. &	Be more careful, it's not a race. &	I'd say it's dry, but there's a bit of doubt. \\
Use a high-tolerance mode for rapid turns. &	Too nervous in steering. &	Seems dry to me, but I can't confirm it. \\
I prefer fast, high-response control with less caution. &	Make it more controlled and less twitchy. &	Looks like a dry surface, but hard to be sure. \\
Push harder on throttle and steering—go all in. &	Reduce the cornering speed. &	Feels like dry conditions, but I might be mistaken. \\
Enable sport mode for maximum power. &	I want safety first, not speed. &	I suspect the road is dry, but I'm unsure. \\
Go fast and be decisive in lane changes. &	It's too reactive, calm it down. &	Probably dry, though I can't guarantee it. \\
Increase steering gain for instant response. &	Go easy on acceleration. &	I think it's dry, though uncertain. \\
I want quick corrections and hard acceleration. &	Too much power, reduce it. &	Dry road, I guess, but not certain. \\
Push for maximum grip and speed. &	Please drive more conservatively. &	The road looks dry, but there's some uncertainty. \\
Drive with racing-level confidence. &	I want slower and safer turns. &	I'd assume it's dry, but not entirely sure. \\
I want it to feel like track mode. &	That was too harsh in handling. &	It might be dry, but I can't confirm. \\
Accelerate rapidly and maintain momentum. &	Reduce the speed margin. &	Seems dry enough, though unsure. \\
Give me immediate throttle response. &	It's unstable, make it cautious. &	Looks pretty dry, but can't be sure. \\
I want to attack the road, not play safe. &	Too strong corrections, smooth them out. &	I believe it's dry, but not completely confident. \\
Go for aggressive lane keeping and acceleration. &	Please tune it for safety and comfort. &	It seems dry overall, but I can't be certain. \\
Turn faster, brake later. &	Too jerky, I want gentle motion. &	The surface appears dry, but I'm not positive. \\
Max performance—no slowdowns. &	Slow down the transitions. &	Probably dry, but I could be wrong. \\
Push through corners with confidence. &	It felt reckless—be more cautious. &	Road seems dry, though I'm not sure. \\
I want hard acceleration out of every turn. &	Reduce lane correction intensity. &	It looks mostly dry, but uncertain. \\
Keep maximum pace throughout the drive. &	I want conservative, not sporty. &	Dry road I think, though with some doubt. \\
Make the system react instantly to inputs. &	Too quick, not steady enough. &	Appears dry, but uncertain. \\
I want extreme handling and control. &	Please lower the gain. &	I feel like the road is dry, but unsure. \\
Drive like a pro racer. &	It's too responsive, I prefer smoothness. &	It looks fine and dry, but who knows. \\
Go strong on speed and agility. &	Too abrupt acceleration. &	Seems like dry conditions, though not sure. \\
I want a forceful, assertive steering behavior. &	Drive more safely, especially on icy roads. &	Might be dry, can't tell for sure. \\
Keep the drive sharp and energetic. &	Make it softer and more stable. &	Looks like dry pavement, though uncertain. \\
Prioritize lap time over comfort. &	It's too quick to react. &	Feels dry, but I can't be 100\% sure. \\
Take the fastest possible path. &	Slow it down, focus on stability. &	It looks dry, but maybe not everywhere. \\
Push vehicle response to the limit. &	Too unsafe—prioritize control. &	The road appears dry overall, but not confirmed. \\
I want fearless control even at high speed. &	Be less confident, more defensive. &	I think it's mostly dry, but can't be certain. \\
Stay aggressive in all situations. &	Reduce the tracking aggressiveness. &	Dry I think, but not completely sure. \\
Increase tolerance for quick maneuvers. &	I prefer a calm driving style. &	Probably dry, but uncertain. \\
I want fast corner entry and exit. &	Too fast in straight lines. &	Seems dry enough, though with doubt. \\
Push dynamic performance higher. &	It should anticipate slower turns. &	It's likely dry, but I'm not certain. \\
Enable the most responsive mode available. &	Please limit the lane error tolerance. &	Looks dry, but I'm hesitant to confirm. \\
I want strong acceleration and tight control. &	Too risky, I want safe behavior. &	I'd say dry, but not confidently. \\
Reduce safety margin for faster results. &	Make it more cautious, less performance-driven. &	It might be dry, but I'm not sure. \\
Push to maximum allowable lane deviation. &	It accelerated too early. &	Feels like a dry surface, but uncertain. \\
I want fast, snappy handling. &	Reduce throttle and steering gain. &	I think it's dry, though some doubt remains. \\
Stay on throttle—no delays. &	I want a more defensive strategy. &	Road looks dry, but who can say for sure. \\
I prefer performance-first control. &	Please be more stable on normal roads. &	It appears dry, but I don't know for certain. \\
Give me top-speed behavior. &	Too rough handling. &	Seems to be dry, though not confirmed. \\
Drive like a sports car, not a family car. &	Make it more balanced, not so sharp. &	Looks quite dry, but can't confirm. \\
Maintain aggressive control throughout. &	It's overreacting to changes. &	I believe it's dry, though with uncertainty. \\
I want to feel constant power and momentum. &	I want smoother adaptation. &	Dry conditions, I think, but maybe not. \\
Drive hard and take initiative. &	Too much correction effort. &	Road surface looks dry, though uncertain. \\
Go faster on every straight segment. &	Slow the vehicle response overall. &	Probably a dry road, but unsure. \\
Push acceleration early in the turn. &	I want predictable, safe control. &	Seems like dry ground, but who knows. \\
React instantly to direction changes. &	Please calm down the acceleration. &	The road appears fine, but I'm not sure it's dry. \\
I want dynamic and energetic motion. &	Reduce speed and keep traction. &	Looks dry enough, but I could be wrong. \\
Keep high momentum and fast transitions. &	I want less risk and smoother response. &	I think the road's dry, though not certain. \\
Make control firm and quick. &	Too aggressive for current conditions. &	Mostly dry, I believe, but not sure. \\
I want aggressive grip management. &	Please adopt a conservative mode. &	It feels dry, but uncertain. \\
Prioritize time over precision. &	It was overconfident in corners. &	Seems like dry asphalt, though uncertain. \\
Go faster with minimal correction delay. &	Slow it down for better control. &	I'd assume it's dry, but with some hesitation. \\
I prefer a daring, confident mode. &	Make it less sensitive to errors. &	Looks to be dry, but can't confirm. \\
Give me maximum responsiveness. &	I want steady, careful movement. &	I think it's fine and dry, but I'm not sure. \\
Increase the pace—make it intense. &	Too fast, not safe on this surface. &	Seems dry, though uncertain. \\
I want minimal damping and high power. &	Reduce control gain significantly. &	Appears to be dry, but who can say for sure. \\
Accelerate sooner and stronger. &	It should react slower to inputs. &	Probably dry, though not guaranteed. \\
Keep speed constant even on curves. &	Please make it safer and smoother. &	Looks dry overall, but I might be wrong. \\
Make steering bold and direct. &	I prefer stable and low-speed driving. &	I believe the road is dry, but unsure. \\
I want powerful, quick decision behavior. &	Too much oscillation, reduce energy. &	Dry, I think, but not confirmed. \\
Push throttle sensitivity higher. &	Slow down in lane keeping. &	Seems like dry weather, though uncertain. \\
Stay fast through all maneuvers. &	I want it to be gentle and cautious. &	Looks fairly dry, but can't be sure. \\
I want performance at the edge of control. &	It was unstable at high speed. &	Probably dry, though uncertain. \\
Drive aggressively, no holding back. &	Too sudden turns—smooth them out. &	I suspect it's dry, but can't confirm. \\
Prioritize speed gains over stability. &	Please reduce lateral aggressiveness. &	Seems to be dry overall, but unsure. \\
I want an assertive, forward-pushing style. &	Go easy—prioritize stability. &	Looks dry at a glance, though uncertain. \\
Use high responsiveness in every control step. &	Too risky on normal roads. &	I think it's dry, but can't guarantee. \\
Keep momentum high even in risk. &	Make it drive more conservatively. &	Dry surface maybe, but not certain. \\
I want more tolerance for deviation. &	It needs to be slower and more composed. &	It looks dry, but not completely sure. \\
Push limits confidently and recover fast. &	I want lower error tolerance. &	I'm guessing it's dry, but not confident. \\
Increase aggressiveness under all conditions. &	Reduce both speed and correction strength. &	Road seems dry, but can't say for sure. \\
Take the apex at full throttle. &	Please drive gently and safely. &	Probably dry, but unsure. \\
Go heavy on speed, light on caution. &	Too sensitive to minor errors. &	Looks like dry pavement, but can't confirm. \\
I want full race-level control. &	Be cautious, like in slippery conditions. &	Seems dry, but could be wrong. \\
Maintain max speed with strong steering. &	I want to keep it slow and safe. &	I think it's dry, though not 100\%. \\
React instantly and push harder each run. &	Too dynamic—please relax the control. &	Dry conditions likely, but uncertain. \\
I prefer competitive driving over safe driving. &	It's unsafe at this pace. &	Road appears dry, but who knows. \\
Keep acceleration high from start to finish. &	Drive slower and smoother overall. &	It looks dry enough, but I'm unsure. \\
Maximize speed and power, ignore comfort. &	Please limit speed and aggression. &	Seems fine and dry, but uncertain. \\
Use aggressive tuning to push limits. &	Too much risk, make it steady. &	Dry I think, but I'm not certain. \\
Make it fast, fearless, and focused. &	I want a relaxed, defensive mode. &	Looks mostly dry, though not sure. \\
\end{longtable}
} 
\vspace{-1em} 
\captionof{table}{Example user instructions describing three distinct driving intentions: aggressive, conservative, and uncertain-surface (dry and unsure).\label{table:user_inputs}}
\vspace{0.5em}

\begin{table}[h!]
\centering
\scriptsize
\resizebox{\textwidth}{!}{%
\begin{tabular}{ll ccc ccc cc}
\toprule
\multicolumn{2}{c}{\textbf{Icy Road: $\mu \in [0.3, 0.4]$}} 
& \multicolumn{3}{c}{Run~1: Aggressive Preference}
& \multicolumn{3}{c}{Run~2: Conservative Preference}
& \multicolumn{2}{c}{Differences} \\[2pt]
\cmidrule(lr){1-2} \cmidrule(lr){3-5} \cmidrule(lr){6-8} \cmidrule(lr){9-10}
\textbf{Method} & \textbf{LLM} 
& \uline{\textit{Lateral}} & \uline{\textit{Speed}} & \uline{\textit{Safety}}
& \uline{\textit{Lateral}} & \uline{\textit{Speed}} & \uline{\textit{Safety}}
& \uline{\textit{Lateral}} & \uline{\textit{Speed}} \\
\midrule\addlinespace[2pt]\midrule
\multirow{3}{*}{APSC} 
& GPT-4o-mini & $0.78 \!\pm\! 0.65$ & $8.17 \!\pm\! 1.54$ & $99\%$
            & $0.55 \!\pm\! 0.38$ & $7.73 \!\pm\! 1.40$ & $100\%$
            & $-0.23$ & $-0.44$ \\
& GPT-3.5 Turbo & $0.88 \!\pm\! 0.51$ & $8.58 \!\pm\! 1.02$ & $99\%$
                 & $0.52 \!\pm\! 0.42$ & $7.94 \!\pm\! 1.40$ & $100\%$
                 & $-0.36$ & $-0.64$ \\

& Gemini 2.5 Flash & $0.90 \!\pm\! 0.54$ & $8.43 \!\pm\! 1.01$ & $99\%$
                    & $0.39 \!\pm\! 0.36$ & $6.92 \!\pm\! 1.52$ & $96\%$
                    & $-0.51$ & $-1.51$ \\

& Gemini 2.0 Flash & $1.03 \!\pm\! 0.68$ & $8.50 \!\pm\! 1.15$ & $96\%$
                    & $0.54 \!\pm\! 0.41$ & $7.55 \!\pm\! 1.57$ & $96\%$
                    & $-0.49$ & $-0.95$ \\
& DeepSeek-Chat & $0.87 \!\pm\! 0.55$ & $8.52 \!\pm\! 1.51$ & $100\%$
                 & $0.43 \!\pm\! 0.41$ & $7.17 \!\pm\! 1.64$ & $100\%$
                 & $-0.44$ & $-1.35$ \\
                 
\midrule
\multirow{3}{*}{AMPC}                  
& GPT-4o-mini & $1.91 \!\pm\! 0.60$ & $8.10 \!\pm\! 0.93$ & $43\%$
               & $1.87 \!\pm\! 0.77$ & $8.46 \!\pm\! 0.95$ & $54\%$
               & $-0.04$ & $+0.36$ \\
& GPT-3.5 Turbo & $2.16 \!\pm\! 0.42$ & $7.90 \!\pm\! 0.72$ & $42\%$
                 & $2.09 \!\pm\! 0.74$ & $8.38 \!\pm\! 0.97$ & $60\%$
                 & $-0.07$ & $+0.48$ \\
& Gemini 2.5 Flash & $1.85 \!\pm\! 0.57$ & $8.20 \!\pm\! 0.99$ & $50\%$
                    & $1.77 \!\pm\! 0.74$ & $8.60 \!\pm\! 1.09$ & $57\%$
                    & $-0.08$ & $+0.40$ \\
& Gemini 2.0 Flash & $1.75 \!\pm\! 0.60$ & $8.00 \!\pm\! 0.74$ & $39\%$
                    & $1.90 \!\pm\! 0.60$ & $7.98 \!\pm\! 0.73$ & $60\%$
                    & $0.15$ & $-0.02$ \\
& DeepSeek-Chat & $2.07 \!\pm\! 0.65$ & $8.15 \!\pm\! 0.86$ & $39\%$
                 & $1.72 \!\pm\! 0.69$ & $8.45 \!\pm\! 1.01$ & $55\%$
                 & $-0.35$ & $+0.30$ \\
\midrule
\multirow{3}{*}{CDBF} 
& GPT-4o-mini & $2.12 \!\pm\! 0.56$ & $8.40 \!\pm\! 1.46$ & $48\%$
               & $1.61 \!\pm\! 0.70$ & $8.32 \!\pm\! 0.79$ & $71\%$
               & $-0.51$ & $-0.08$ \\
& GPT-3.5 Turbo & $1.96 \!\pm\! 0.60$ & $9.88 \!\pm\! 5.52$ & $64\%$
                 & $1.67 \!\pm\! 0.70$ & $8.38 \!\pm\! 0.79$ & $61\%$
                 & $-0.29$ & $-1.50$ \\
& Gemini 2.5 Flash & $1.62 \!\pm\! 0.69$ & $9.31 \!\pm\! 4.45$ & $68\%$
                    & $1.52 \!\pm\! 0.59$ & $9.03 \!\pm\! 1.75$ & $66\%$
                    & $-0.10$ & $-0.28$ \\
& Gemini 2.0 Flash & $1.76 \!\pm\! 0.71$ & $8.78 \!\pm\! 1.49$ & $57\%$
                    & $1.64 \!\pm\! 0.59$ & $8.63 \!\pm\! 1.56$ & $69\%$
                    & $-0.12$ & $-0.15$ \\
& DeepSeek-Chat & $2.09 \!\pm\! 0.54$ & $8.18 \!\pm\! 1.39$ & $52\%$
                 & $1.60 \!\pm\! 0.65$ & $8.65 \!\pm\! 1.53$ & $72\%$
                 & $-0.49$ & $+0.47$ \\
\bottomrule
\end{tabular}}
\caption{Full quantitative results for case (i) under icy road conditions.}
\label{tab:llm_case_1_full_icy}
\end{table}

\begin{table}[h!]
\centering
\scriptsize
\resizebox{\textwidth}{!}{%
\begin{tabular}{ll ccc ccc cc}
\toprule
\multicolumn{2}{c}{\textbf{Icy Road: $\mu \in [0.3, 0.4]$}} 
& \multicolumn{3}{c}{Run~1: Wrong Premise on Dry Road}
& \multicolumn{3}{c}{Run~2: Wrong Premise on Dry Road}
& \multicolumn{2}{c}{Differences} \\[2pt]
\cmidrule(lr){1-2} \cmidrule(lr){3-5} \cmidrule(lr){6-8} \cmidrule(lr){9-10}
\textbf{Method} & \textbf{LLM} 
& \uline{\textit{Lateral}} & \uline{\textit{Speed}} & \uline{\textit{Safety}} 
& \uline{\textit{Lateral}} & \uline{\textit{Speed}} & \uline{\textit{Safety}} 
& \uline{\textit{Lateral}} & \uline{\textit{Speed}} \\ 
\midrule\addlinespace[2pt]\midrule
\multirow{3}{*}{APSC} 
& GPT-4o-mini & $0.82 \!\pm\! 0.16$ & $0.35 \!\pm\! 0.03$ & $94\%$
            & $ 0.58 \!\pm\! 0.21$ & $0.35 \!\pm\! 0.03$ & $100\%$
            & $-0.14$ & $-0.32$ \\ 
& GPT-3.5 Turbo
& $0.90 \!\pm\! 0.04$ & $0.35 \!\pm\! 0.03$ & $95\%$
& $0.44 \!\pm\! 0.24$ & $0.35 \!\pm\! 0.03$ & $100\%$
& $-0.30$ & $-0.11$ \\
            
& Gemini 2.5 Flash   & $0.90 \!\pm\! 0.04$ & $0.36 \!\pm\! 0.03$ & $100\%$
            & $0.56 \!\pm\! 0.26$ & $0.36 \!\pm\! 0.03$ & $100\%$
            & $-0.11$ & $-0.08$ \\ 
            
& Gemini 2.0 Flash
& $0.89 \!\pm\! 0.06$ & $0.36 \!\pm\! 0.03$ & $97\%$
& $0.47 \!\pm\! 0.20$ & $0.35 \!\pm\! 0.03$ & $100\%$
& $-0.12$ & $-0.12$ \\

& DeepSeek-Chat & $0.90 \!\pm\! 0.04$ & $0.35 \!\pm\! 0.03$ & $98\%$
            & $0.84 \!\pm\! 0.14$ & $0.36 \!\pm\! 0.03$ & $100\%$
            & $-0.42$ & $+0.04$ \\ 

\midrule
\multirow{3}{*}{AMPC} 
& GPT-4o-mini
& $0.82 \!\pm\! 0.16$ & $0.35 \!\pm\! 0.03$ & $23\%$
& $0.44 \!\pm\! 0.12$ & $0.35 \!\pm\! 0.03$ & $56\%$
& $-0.30$ & $-0.55$ \\
& GPT-3.5 Turbo
& $0.90 \!\pm\! 0.04$ & $0.35 \!\pm\! 0.03$ & $38\%$
& $0.40 \!\pm\! 0.20$ & $0.36 \!\pm\! 0.03$ & $65\%$
& $+0.23$ & $-0.16$ \\
& Gemini 2.5 Flash
& $0.90 \!\pm\! 0.04$ & $0.35 \!\pm\! 0.03$ & $78\%$
& $0.70 \!\pm\! 0.28$ & $0.34 \!\pm\! 0.03$ & $63\%$
& $+0.04$ & $+0.10$ \\
& Gemini 2.0 Flash
& $0.89 \!\pm\! 0.06$ & $0.35 \!\pm\! 0.03$ & $41\%$
& $0.45 \!\pm\! 0.19$ & $0.35 \!\pm\! 0.03$ & $49\%$
& $+0.07$ & $-1.03$ \\
& DeepSeek-Chat
& $0.90 \!\pm\! 0.04$ & $0.35 \!\pm\! 0.03$ & $28\%$
& $0.60 \!\pm\! 0.17$ & $0.34 \!\pm\! 0.03$ & $62\%$
& $-0.57$ & $-0.09$ \\
 
\midrule
\multirow{3}{*}{CDBF} 
& GPT\textendash 4o\textendash mini
& $0.82 \!\pm\! 0.16$ & $0.34 \!\pm\! 0.03$ & $21\%$
& $0.43 \!\pm\! 0.10$ & $0.35 \!\pm\! 0.03$ & $63\%$
& $-0.82$ & $-0.05$ \\
& GPT-3.5 Turbo
& $0.90 \!\pm\! 0.04$ & $0.34 \!\pm\! 0.03$ & $32\%$
& $0.37 \!\pm\! 0.17$ & $0.36 \!\pm\! 0.02$ & $76\%$
& $-0.59$ & $-0.11$ \\
& Gemini 2.5 Flash
& $0.90 \!\pm\! 0.04$ & $0.36 \!\pm\! 0.03$ & $59\%$
& $0.69 \!\pm\! 0.28$ & $0.36 \!\pm\! 0.03$ & $56\%$
& $+0.07$ & $-1.10$ \\
& Gemini 2.0 Flash
& $0.89 \!\pm\! 0.06$ & $0.35 \!\pm\! 0.03$ & $35\%$
& $0.45 \!\pm\! 0.22$ & $0.35 \!\pm\! 0.03$ & $51\%$
& $-0.21$ & $-1.13$ \\
& DeepSeek\textendash Chat
& $0.90 \!\pm\! 0.04$ & $0.35 \!\pm\! 0.04$ & $19\%$
& $0.60 \!\pm\! 0.17$ & $0.35 \!\pm\! 0.03$ & $45\%$
& $-0.08$ & $+0.17$ \\

\bottomrule
\end{tabular}}
\caption{Full quantitative results for case (ii).}
\label{tab:llm_case_2_full}
\end{table}

\section{Comparison with Non-adaptive Controllers}
\label{sec:exp_verification}

In this section, we show that the proposed adaptive probabilistic safety certificate effectively mediates the trade-off between safety and performance in uncertain road conditions against non-adaptive controllers. 
For the nominal controller $\pi$, we consider the following state-feedback controller for linear vehicle models as in~\cite{ames2017control}.
\begin{align}
 \pi(x) =  - K (x - x_\mathrm{ff}(x)),  \label{eq:nominal_controller}
\end{align}
where $K$ is the feedback gain and $x_\mathrm{ff}$ is a feed-forward term. The control inputs $\Delta \delta$ and $\Delta\tau_\mathrm{e}$ are 
\begin{align}
    \Delta \delta = K_\mathrm{lateral}( [v_y,\, r, \, \delta, e, \psi ]^\top - x_\mathrm{ff,lateral} ),
\end{align}
with $x_\mathrm{ff,lateral}=[0, v_x \rho(s), 0, 0, 0]^\top$ and $K_\mathrm{lateral}$ determined by solving an LQR~\citep{ames2017control}, and 
\begin{align}
    \Delta \tau_\mathrm{e} = -K_\mathrm{v}(v_x - V_\mathrm{ref}) - K_\mathrm{T} T_\mathrm{e},
\end{align}
where the reference $V_\mathrm{ref}$ is set to \SI{40}{km/h}, and the gains $K_\mathrm{v}$ and $K_\mathrm{T}$ are tuned such that the vehicle speed converges to the reference value by \SI{30}{s} on a straight road. 
Thus, the gain $K$ and the term $x_\mathrm{ff}$ in \eqref{eq:nominal_controller} follow by properly aligning $K_\mathrm{lateral}$, $K_\mathrm{v}$, $K_\mathrm{T}$, $V_\mathrm{ref}$, and $v_x\rho(s)$. 
Then the safe control action is given by solving the safe control optimization~\eqref{eq:conditioning} with objective $J = \dfrac{1}{2} \| \tilde{u} - \pi(X_k) \|^2$, 
where $\pi$ is the reference controller $N$ in \eqref{eq:conditioning}, and the risk tolerance $\epsilon$ and the function $\gamma$ in \eqref{eq:safety_condition_each_Zt_aug} are $\epsilon = 0.1$ and $\gamma(x) = x$. 
The safety probability $\Psi^\pi_{H_k}(X_k)$ and its derivatives are estimated by Monte Carlo simulations with a sampling rate of \SI{0.1}{s} and an outlook horizon of $T=75$ (i.e., \SI{7.5}{s} ahead) with 100 samples. 
Note that this computation can be replaced by other online/offline methods to learn the safety probability from driving data or numerical simulations~\citep{wang2023generalizable,wang2025generalizable,hoshino2024physics}. In this example, the estimator prior at $k=0$ is chosen as $\mu_0=0.3$ and $\sigma_0^2=0.01$, and $\bar{\sigma}^2$ is assumed to be 0.1.  
Figure~\ref{fig:posterior_distribution} shows how the posterior distribution changes over time in a case where the true friction coefficient is $\mu^\ast=0.7$. It can be seen that the prior covers the icy road conditions, and the distribution gradually shifts to the right over time. 

\begin{figure}[t!]
 \centering
\includegraphics[width=0.8\linewidth]{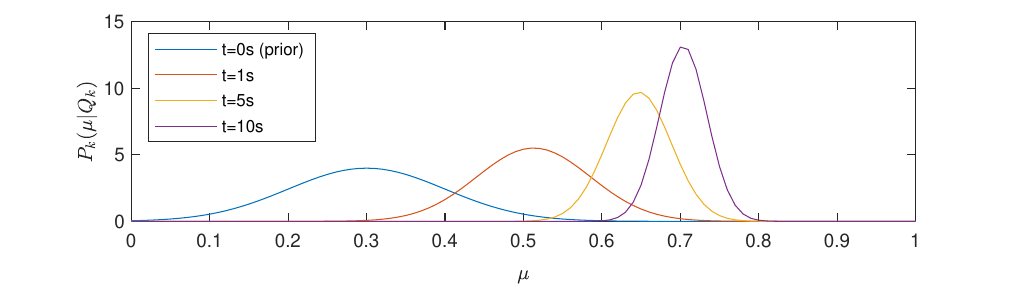} 
 \caption{Sequential update of the posterior distribution of $\mu$ when the true friction coefficient is $\mu^\ast = 0.7$. } 
 \label{fig:posterior_distribution}
\end{figure}

\begin{figure}[t!]
    \centering

    \begin{minipage}[b]{0.30\linewidth}
        \centering
        \includegraphics[width=\linewidth]{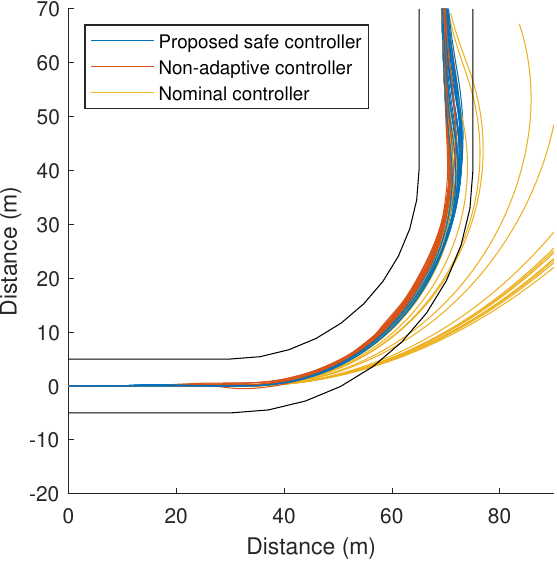}
        \caption{Vehicle trajectories.}
        \label{fig:vehicle_trajectories}
    \end{minipage}
    \hfill
    \begin{minipage}[b]{0.32\linewidth}
        \centering
        \includegraphics[width=\linewidth]{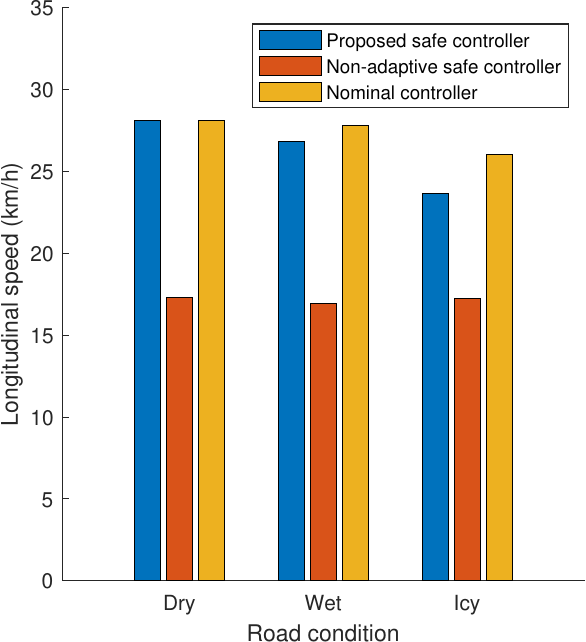}
        \caption{Longitudinal speed.}
        \label{fig:speed_vs_roadcond}
    \end{minipage}
    \hfill
    \begin{minipage}[b]{0.30\linewidth}
        \centering
        \includegraphics[width=\linewidth]{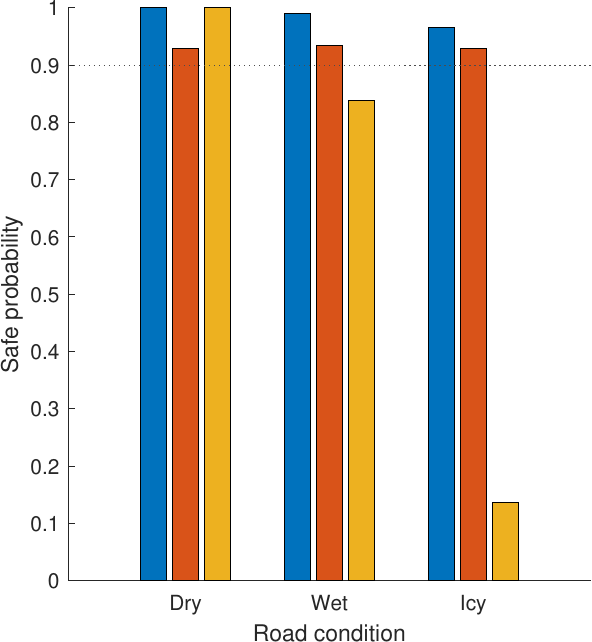}
        \caption{Safety probability.}
        \label{fig:safe_vs_roadcond}
    \end{minipage}


\end{figure}

We compare the closed-loop performance of the proposed adaptive safe controller with two cases: (i) using the nominal controller $\pi$ directly (\textbf{Nominal controller}), and (ii) using the proposed method without updating the posterior $P_k(\mu|Q_k)$ (\textbf{Non-adaptive safe controller}). 
Figure~\ref{fig:vehicle_trajectories} shows the vehicle trajectories for 30 scenarios with different values of the friction coefficient $\mu$, with 10 trajectories for each of the three road conditions: dry (0.7--0.9), wet (0.4--0.7), and icy (0.2--0.4) road conditions. 
We observe that with the proposed safe controller and the non-adaptive safe controller, the vehicle always stays in the lane and turns the curve safely in all cases, but the nominal controller fails to keep the vehicle safe because it does not account for tire-force saturation. 
Figure~\ref{fig:speed_vs_roadcond} and~\ref{fig:safe_vs_roadcond} summarize the trade-offs between safety and performance by showing the longitudinal velocity $v_x$ and the safety probability $\Psi^\pi_{H_k}$, averaged over time and 10 scenarios for each road condition. 
When the nominal controller is used, the longitudinal velocity is high, but the safety probability is low for wet or icy conditions. 
In contrast, the non-adaptive safe controller achieves a high safety probability but keeps the vehicle speed low regardless of the road condition, exemplifying the over-conservatism of worst-case approaches.  
The proposed adaptive safe control method achieves higher longitudinal velocities while ensuring safety by exploiting the posterior updates based on online estimation of the friction coefficient.

\section{LLM Implementation Details}
\label{sec:appendix_llm_implementation}

In this section, we provide additional details on LLM designs and implementation. Both the inference LLM and reasoning LLM are LLMs with specific system-level prompts, to take natural language inputs on control preference to specific control executables. Specifically, for the experiment settings in Section~\ref{sec:experiment}, we use the system-level prompts for inference and reasoning as below. 
The overall procedures for LLM integration with adaptive PSC is shown in Alg.~\ref{alg:llm_loop}.

\begin{algorithm}[H]
\caption{LLM Integration with Adaptive PSC}\label{alg:llm_loop}
\SetAlgoLined
\DontPrintSemicolon

\KwIn{Environment $\mathcal{M}$, total runs $N$, system prompt $P_{\text{sys}}$, LLM model $\mathcal{L}$, human instructions $\{I_n\}_{n=1}^N$}
\KwOut{Rationales $\{R_n\}_{n=1}^N$, control executables $\{E_n\}_{n=1}^N$, logs $\{D_n\}_{n=1}^N$}

\textbf{LLM setup:} Load model $\mathcal{L}$ with $P_{\text{sys}}$ and
initialize buffers: $I_{1:0}\gets\varnothing$, $D_{1:0}\gets\varnothing$.\;

\SetKwFunction{LLMPlan}{LLMPlan}
\SetKwProg{Fn}{Function}{:}{}
\Fn{\LLMPlan{$\mathcal{L}$,\, $I_{1:n},\, D_{1:n-1}$}}{
  Query LLM $\mathcal{L}$ with $\big(P_{\text{sys}},\ I_{1:n},\ D_{1:n-1}\big)$ to obtain:\;
  {\setlength{\leftmargini}{10pt}%
   \setlength{\itemsep}{-5pt}%
   \setlength{\parskip}{-2pt}%
   \begin{itemize}
     \item \textbf{Rationales} $R_n$ \tcp*{interpretable reasoning traces}
     \item \textbf{Executables} $E_n=\{\phi_n,~\mathcal{E}_n,~J_n~,~\cdots\}$ \tcp*{safety, estimator \& control specifications}
   \end{itemize}}%
  \KwRet $E_n$\;
}

\For{$\mathrm{Run}\; n \gets 1$ \KwTo $N$}{
  Human input instruction $I_n$; \quad $I_{1:n}\gets I_{1:n-1}\cup\{I_n\}$\;

  \eIf{$n=1$}{
    $R_n, E_n \gets$ \LLMPlan{$\mathcal{L}$,\, $I_{1:1},\, D_{1:0}$} \tcp*{no prior data}
  }{
    $R_n, E_n \gets$ \LLMPlan{$\mathcal{L}$,\, $I_{1:n},\, D_{1:n-1}$}
  }

  Output $R_n$ as a summary\;
  
  Perform adaptive safe control with $E_n$ and produce trajectory $\tau_n$ and metrics\; \tcp*{e.g., Adaptive PSC in Alg.~\ref{alg:safe_control}}

  Form log $D_n=\{\tau_n,\ \text{safety stats},\ \text{costs},\ \text{posterior from }\mathcal{E}_n,~\cdots\}$;\;
  $D_{1:n}\gets D_{1:n-1}\cup\{D_n\}$\;
}
\end{algorithm}


\vspace{1em}

\textbf{Inference LLM system prompts:}

\begin{lstlisting}
"You are an expert inferring initial controller priors from a short user preference." + ...
    " Return STRICT JSON with keys: e_max, mu_0, sigma_0, bar_sigma, assumptions, rationale." + ...
    " Meanings:" + ...
    " - e_max: maximum lane tracking error tolerance. Larger for more aggressive turns/risk; smaller for more conservative/precise." + ...
    " - mu_0: initial prior for road-tire friction (icy small, normal medium, dry large)." + ...
    " - sigma_0: uncertainty (std^2) of the friction prior; larger if the user sounds unsure or contradictory." + ...
    " - bar_sigma: confidence of the estimator on its measurements; increase if not sure if estimator is good." + ...
    " Policy:" + ...
    " - If the user uses vague words (""seems"", ""maybe"", ""not sure"", ""probably""), pick the most likely road class they stated," + ...
    " - Only change e_max with explicit user cues." + ...
    " Valid discrete ranges:" + ...
    " - e_max in {3,5,10}; mu_0 in {0.3,0.5,0.9}; sigma_0 in {0.05,0.3}; bar_sigma in {0.05,0.3}." + ...
    " Output ONLY JSON in this exact shape: " + ...
    " {""e_max"":0,""mu_0"":0.0,""sigma_0"":0.0,""bar_sigma"":0.0,""assumptions"":{""style"":"""",""road"":"""",""speed_kmh"":0,""lane_quality"":""""},""rationale"":""""} " + ...
    " Ensure values are from the allowed sets and remember these are initial priors, not ground truth."
\end{lstlisting}

\textbf{Reasoning LLM system prompts:}

\begin{lstlisting}
    "=== Quantitative feedback from the last run ===\n%s\n\n" + ...
    "=== User feedback ===\n%s\n\n" + ...
    "=== Your task ===\n" + ...
    "Based on both quantitative feedback and qualitative user feedback, infer new reasonable parameters." + ...
    "Return STRICT JSON with keys: e_max, mu_0, sigma_0, bar_sigma, assumptions, rationale.\n" + ...
    " Meanings:" + ...
    " - e_max: maximum lane tracking error tolerance. Larger for more aggressive turns/risk; smaller for more conservative/precise." + ...
    " - mu_0: initial prior for road-tire friction (icy small, normal medium, dry large)." + ...
    " - sigma_0: uncertainty (std^2) of the friction prior; larger if the user sounds unsure or contradictory." + ...
    " - bar_sigma: confidence of the estimator on its measurements; increase if not sure if estimator is good. You should try to trust the estimator." + ...
    " Policy:" + ...
    " - If the user uses vague words (""seems"", ""maybe"", ""not sure"", ""probably""), pick the most likely road class they stated," + ...
    " - Only change e_max with explicit user cues." + ...
    " Valid discrete ranges:" + ...
    " - e_max in {3,5,10}; mu_0 in {0.3,0.5,0.9}; sigma_0 in {0.05,0.3}; bar_sigma in {0.05,0.3}." + ...
    " - Keep bar_sigma=0.05. Only set bar_sigma=0.3 if the text explicitly mentions sensing/visibility problems (fog/rain/snow/glare/low light/sensor fault) and explain why.\n" + ...
    " - Keep sigma_0=0.05 and align mu_0 with previous estimation. Set sigma_0=0.3 and different mu_0 only if (a) the user hedges about the ROAD (""seems/maybe/not sure/probably""), or (b) the user statement contradicts quantitative feedback suggesting a different friction class; explain the uncertainty.\n" + ...
    " Output ONLY JSON in this exact shape: " + ...
    " {""e_max"":0,""mu_0"":0.0,""sigma_0"":0.0,""bar_sigma"":0.0,""assumptions"":{""style"":"""",""road"":"""",""speed_kmh"":0,""lane_quality"":""""},""rationale"":""""} " + ...
    "Ensure values are from the allowed sets and remember these are initial priors, not ground truth.\n";
\end{lstlisting}

\section{LLM for Barrier Function Generation}
\label{sec:llm_generated_cbf}

\begin{figure}
    \centering
    \includegraphics[width=0.85\linewidth]{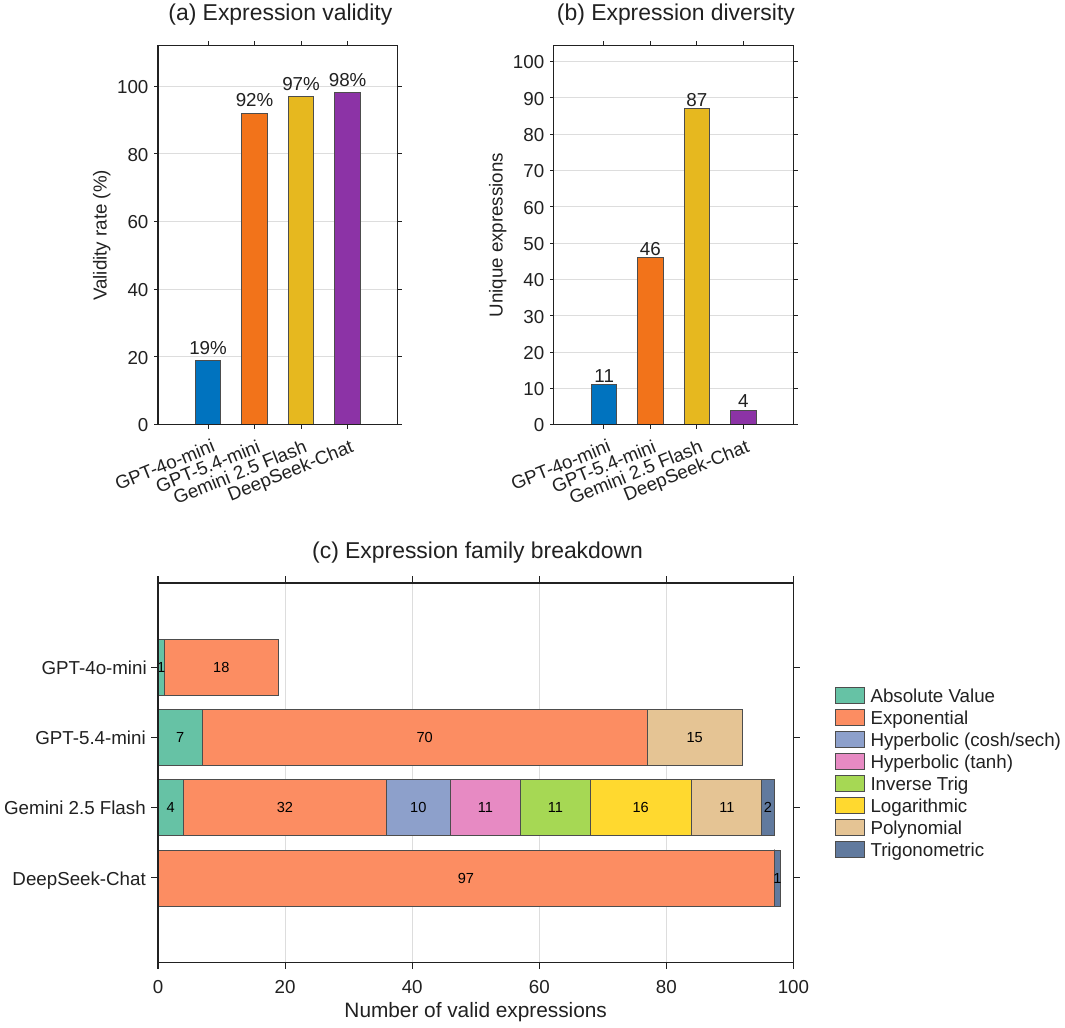}
    \caption{Validity, diversity and form breakdown of LLM generated barrier functions.}
    \label{fig:llm_generated_phi}
\end{figure}

\begin{figure}
    \centering
    \includegraphics[width=0.55\linewidth]{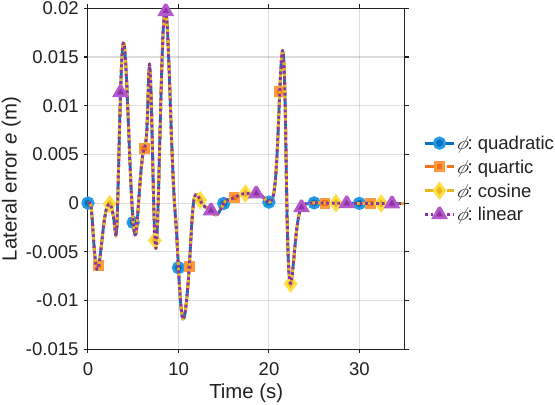}
    \caption{Lateral lane error under the proposed adaptive PSC using different LLM-generated barrier functions for a fixed safety requirement. All trajectories overlap, indicating that the method is agnostic to the barrier function form.}
    \label{fig:psc_results_llm_phi}
\end{figure}

\begin{figure}
    \centering
    \includegraphics[width=\columnwidth]{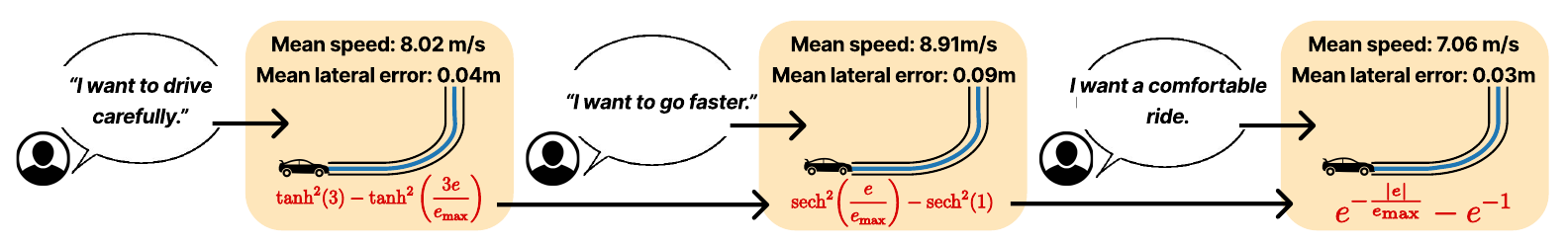}
    \caption{Qualitative results across three rounds of human instructions and feedback with LLM generated barrier functions $\phi$ (shown in red).}
    \label{fig:llm_phi_qualitative}
\end{figure}

In this section, we present additional results demonstrating that the proposed framework can accommodate a diverse set of barrier functions, including those generated by LLMs, without degrading safety performance.

Recall that in Section~\ref{sec:experiment}, we used a fixed barrier function~\eqref{eq:lane_error_safe_set} to characterize the lane-tracking safety requirement. Here, we conduct additional experiments to show that the proposed method can accommodate a wide range of barrier function forms for the same safety specification, without affecting safety performance. This is because the framework operates in the probability space and is agnostic to the specific form of the barrier function.

We first demonstrate that LLMs are capable of generating diverse candidate barrier functions. Fig.~\ref{fig:llm_generated_phi} summarizes the validity, expression diversity, and functional form distribution of the generated barrier functions. Notably, although several example forms are provided in the prompt, none of the generated functions replicate those examples, indicating generative diversity of the tested LLMs. We also observe that more advanced LLMs tend to produce barrier functions with higher validity and diversity. We expect validity to improve further as LLM capabilities continue to advance, alongside improved prompt design. Examples of LLM-generated barrier functions for the lane error safety requirement considered in this paper are shown below.
\begin{align*}
\phi_1(e) &= \exp\!\left(-\frac{|e|}{e_{\max}}\right) - 1  &
\phi_{2}(e) &= \frac{e_{\max}^2 - e^2}{e_{\max}^2 + e^2} \\
\phi_3(e) &= \exp\!\left(-\frac{e^2}{e_{\max}^2}\right) - e^{-1} &
\phi_4(e) &= 1 - \left(\frac{|e|}{e_{\max}}\right)^3
\end{align*}

Next, we show that using these LLM-generated barrier functions, the proposed adaptive PSC framework yields consistent behavior in ensuring long-term safety. Specifically, we evaluate the controller using four distinct barrier functions generated by the LLM for the lane-keeping task, and report the resulting lateral lane error in Fig.~\ref{fig:psc_results_llm_phi}. Despite differences in functional form, all barrier functions encode the same safety requirement, and the resulting system behavior is identical across all cases. This highlights a key advantage of our approach that by operating in the probability space, it remains invariant to the particular choice of barrier function representation.

Finally, we demonstrate that LLM-based barrier function generation can be seamlessly integrated into the adaptive PSC framework in a closed-loop setting. We conduct a multi-turn driving simulation in which the LLM generates barrier functions from scratch, and the adaptive PSC uses them for safe control. Fig.~\ref{fig:llm_phi_qualitative} presents examples of the resulting trajectories and key performance statistics. The results show that the LLM consistently generates valid barrier functions along with compatible estimator configurations, enabling safe driving behavior that aligns with human preferences while maintaining long-term safety guarantees.

All key prompts used in the experiments presented in this section are provided below. 
The full simulation code is available in our GitHub repository.

\newpage

\begin{lstlisting}
// Generating barrier functions

Produce a unique safety barrier function expression in MATLAB syntax using variables 'e' (lateral error) and 'emax' (maximum allowed error; use the same numeric scale as your chosen e_max). The function phi(e, emax) must satisfy: 1) phi(0, emax) >= 0 (usually 1) 2) phi(emax, emax) <= 0 3) phi(-emax, emax) <= 0 Here are some standard working examples: '1-(e/emax)^2', '1-(e/emax)^4', 'cos(pi*e/(2*emax))', or '1-abs(e/emax)'. You may use the examples or return highly varied shapes that still satisfy the constraints (higher-order polynomials, fractional powers where defined, sech/tanh, bounded exponentials, squishing functions, etc.).

Return ONLY JSON in the exact shape: {"phi_expr": "your_expression_here"}

// Integration with adaptive PSC

- phi_expr: You are a mathematical control theory expert designing the PSC barrier for this user.

You are an expert assigning reasonable initial guesses for parameters in an adaptive lane-keeping controller to ensure safety and efficiency.

You are a mathematical control theory expert updating the PSC barrier using the quantitative feedback and user message.

\end{lstlisting}

\end{document}